\documentclass[aps,onecolumn,a4paper,superscriptaddress,longbibliography,nofootinbib,notitlepage,groupedaddress]{revtex4-1}

\usepackage{lipsum}
\usepackage{setspace}

\usepackage{amsmath}
\usepackage{amssymb}
\usepackage{amsthm}
\usepackage[utf8]{inputenc}
\usepackage{comment}
\usepackage{color}

\usepackage[font=scriptsize]{caption}
\usepackage{mwe}

\newcommand{\abs}[1]{\left\lvert#1\right\rvert} 
\newcommand{\norm}[1]{\left\lVert#1\right\rVert} 

\usepackage[a4paper,total={6.5in,9in}]{geometry}
\usepackage[english]{babel}
\usepackage[T1]{fontenc}
\usepackage{mathrsfs}
\usepackage{braket}
\usepackage{physics}
\usepackage{enumerate}
\usepackage{graphicx,float}
\usepackage{mathtools}
\usepackage{microtype}

\usepackage{adjustbox}

\usepackage{caption}
\theoremstyle{plain}

\usepackage[colorlinks = true,
            linkcolor = red,
            urlcolor  = blue,
            citecolor = blue,
            anchorcolor = red]{hyperref}

\newtheorem{theorem}{Theorem}
\newtheorem{corollary}[theorem]{Corollary}

\newtheorem{lemma}[theorem]{Lemma}
 
\newtheorem{definition}[theorem]{Definition}  
\newtheorem{remark}[theorem]{Remark}  
\newtheorem*{remark*}{Remark}   
\renewcommand\qedsymbol{$\blacksquare$}
\newenvironment{proof-of}[1][{\hspace{-\blank}}]{{\medskip\noindent\textit{Proof~{#1}.\ }}}{\hfill\qedsymbol}

\renewcommand{\Tr}{{\operatorname{Tr}\,}}
\renewcommand{\Pr}{{\operatorname{Pr}}}

\newcommand{\id}{{\operatorname{id}}}

\newcommand{\1}{\openone}
\newcommand{\proj}[1]{|#1\rangle\!\langle #1|}
\newcommand{\cD}{{\mathcal{D}}}
\newcommand{\cE}{{\mathcal{E}}}
\newcommand{\cT}{{\mathcal{T}}}

\newcommand{\KI}{{\text{KI}}}

\newcommand{\nc}{\newcommand}
\nc{\rnc}{\renewcommand}
\nc{\avg}[1]{\langle#1\rangle}
\nc{\Rank}{\operatorname{Rank}}
\nc{\smfrac}[2]{\mbox{$\frac{#1}{#2}$}}

\nc{\ox}{\otimes}
\nc{\dg}{\dagger}
\nc{\dn}{\downarrow}
\nc{\cA}{{\cal A}}
\nc{\cB}{{\cal B}}
\nc{\cC}{{\cal C}}
\nc{\cF}{{\cal F}}
\nc{\cG}{{\cal G}}
\nc{\cH}{{\cal H}}
\nc{\cI}{{\cal I}}
\nc{\cJ}{{\cal J}}
\nc{\cK}{{\cal K}}
\nc{\cL}{{\cal L}}
\nc{\cM}{{\cal M}}
\nc{\cN}{{\cal N}}
\nc{\cO}{{\cal O}}
\nc{\cP}{{\cal P}}
\nc{\cQ}{{\cal Q}}
\nc{\cR}{{\cal R}}
\nc{\cS}{{\cal S}}
\nc{\cW}{{\cal W}}
\nc{\cX}{{\cal X}}
\nc{\cY}{{\cal Y}}
\nc{\cZ}{{\cal Z}}
\nc{\csupp}{{\operatorname{csupp}}}
\nc{\qsupp}{{\operatorname{qsupp}}}
\nc{\rar}{\rightarrow}
\nc{\lrar}{\longrightarrow}
\nc{\polylog}{{\operatorname{polylog}}}
\nc{\wt}{{\operatorname{wt}}}
\nc{\DS}{{\text{DS}}}

\nc{\RR}{{{\mathbb R}}}
\nc{\CC}{{{\mathbb C}}}
\nc{\FF}{{{\mathbb F}}}
\nc{\NN}{{{\mathbb N}}}
\nc{\ZZ}{{{\mathbb Z}}}
\nc{\PP}{{{\mathbb P}}}
\nc{\QQ}{{{\mathbb Q}}}
\nc{\UU}{{{\mathbb U}}}
\nc{\EE}{{{\mathbb E}}}

\nc{\Hom}[2]{\mbox{Hom}(\CC^{#1},\CC^{#2})}
\nc{\rU}{\mbox{U}}

\nc{\ob}[1]{#1}

\nc{\SEP}{{\text{SEP}}}
\nc{\NS}{{\text{NS}}}
\nc{\LOCC}{{\text{LOCC}}}
\nc{\PPT}{{\text{PPT}}}
\nc{\EXT}{{\text{EXT}}}
\nc{\Sym}{{\operatorname{Sym}}}

\nc{\ERLO}{{E_{\text{r,LO}}}}
\nc{\ERLOCC}{{E_{\text{r,LOCC}}}}
\nc{\ERPPT}{{E_{\text{r,PPT}}}}
\nc{\ERLOCCinfty}{{E^{\infty}_{\text{r,LOCC}}}}
\nc{\Aram}{{\operatorname{\sf A}}}

\begin{document}

\title{\LARGE{The Generalized Capacity of a Quantum Channel}}

\author{Zahra Baghali Khanian}
\affiliation{
Department of Mathematics, Technical University of Munich, 85748 Garching, Germany \\
\& Munich Center for Quantum Science and Technology 
}

\begin{abstract}
The transmission of classical information over a classical channel gave rise to the classical capacity theorem with the optimal rate in terms of the classical mutual information.
Despite classical information being a subset of quantum information, 
the rate of the quantum capacity problem is expressed in terms of the coherent information,
which does not mathematically generalize the classical mutual information. Additionally,  
there are multiple capacity theorems with distinct formulas when dealing with  transmitting information over a noisy quantum channel.  
This leads to the question of what constitutes a mathematically accurate quantum generalization of classical mutual information and whether there exists a quantum task that directly extends the classical capacity problem.
In this paper, we address these inquiries by introducing a quantity called the generalized information, which serves as a mathematical extension encompassing both classical mutual information and coherent information. We define a transmission task, which includes as specific instances both classical information and quantum information capacity problems, and show that the transmission capacity of this task is characterized by the generalized information. 
%
%
%
%
%
%
%
\end{abstract}
  
\maketitle
\section{Introduction and Background}

In the 90's people started to look for a quantum analogue of the classical mutual information
and a capacity problem which could give an operational meaning to this  quantity. 
Shor, in  his seminal paper \cite{decoherence-Shor-1995}, 
mentioned that quantum error correction is a step toward the quantum analog of channel
coding in classical information theory.
%
%
Then, in a series of papers the quantum analogue of the classical capacity theorem was
defined as transmitting entanglement, or equivalently a subspace, over a quantum channel, and
the coherent information was established as an upper bound on the quantum capacity \cite{Schumi1996,err-correction-Schumacher-1996,Barnum1998,fidelities-capacities-2000}, where
it was argued that even though the coherent information does not mathematically generalize the classical mutual information, since it appears in the entanglement transmission task and has similar properties such as obeying the data-processing inequality, it is a plausible candidate for the quantum analogue of the classical mutual information. 
The coherent information was established as a lower bound on the quantum capacity problem in \cite{Lloyd_capacity_97,Shor_direct_capacity2002,Devetak-capacity-2005}.

However,  this question still remained what a mathematically correct quantum generalization of the classical mutual information is, and whether there is quantum task which literally
generalizes the classical capacity problem? In this paper, we answer these questions
by introducing a new quantity in Definition~\ref{def: I_G}, which we call it the generalized information, and define a quantum task that not only does retrieve the classical and quantum capacity problems as special cases, it unifies them, despite them have been considered as distinct problems so far.
The generalized information, which is a  combination of the coherent information and the mutual information, was not  officially defined before, but it was previously observed by Devetak and Shor in \cite{Devetak-Shor-2005}.
They found an intriguing connection between the capacity region in \cite{Devetak-Shor-2005} and the findings of Shor in \cite{c_capacity_Shor2004} concerning the classical capacity of a quantum channel with limited entanglement assistance. The connection was that the sum of the classical and  quantum (or entanglement) rates in both tasks were the same, and it was equal to a quantity which we call it the generalized information.

The task, which gives operational meaning to the generalized information, is about transmitting 
${A'}^m$ part of $m$ copies of a general mixed state $\rho^{A'R}$ as shown in Fig~\ref{fig: diagram}. This task is related to simultaneous transmission of classical and quantum information defined in \cite{Devetak-Shor-2005} in that the correlations of system $A'$ with the reference system $R$ can be decomposed into classical and quantum correlations due to the Koashi-Imoto decomposition of a state $\rho^{A'R}$ \cite{KI2001,KI2002,Hayden2004}. 
Here, we briefly review this simultaneous transmission task. 

%


We first remind that the coherent information of a channel $\cN:A \to B$ for a state $\rho^{A}$ is defined as \cite{err-correction-Schumacher-1996}
\begin{align}
    I_c(\rho^{A},\cN):=I(R \> \rangle B )_{\sigma}=-S(R|B)_{\sigma}=S(B)_{\sigma}-S(BR)_{\sigma},
\end{align}
where the  quantities are in terms of the state $\sigma^{BR}=(\cN\ox \id_R)\ketbra{\rho}^{AR}$, and $\ket{\rho}^{AR}$ is a purification of $\rho^{A}$.
Also, the Holevo information of a channel $\cN:A \to B$ for a cq state $\rho^{AX}=\sum_x p_x \rho_x^A \ox \proj{x}^X$ is defined as \cite{Holevo1973}
\begin{align}
    \chi(\rho^{AX},\cN):=I(B:X)_{\tau},
\end{align}
where the quantum mutual information is in terms of the state $\tau^{BX}=(\cN\ox \id_X) {\rho}^{AX}$.

The simultaneous transmission of classical and quantum information is 
a communication task considered in \cite{Devetak-Shor-2005}. 
The aim of this task is to  transmit system $A'$ of is a maximally entangled state 
$\ket{\Phi_{\kappa}}^{A'R}$ of dimension $\kappa$, and a classical message $M$ from a set 
 $\{1,2,\cdots,\mu\}$.
An $(n,\epsilon)$ code for this task consists of an encoder $\cE_m: A' \to A^n$,
a CPTP map depending on a message $m$, and 
a quantum instrument decoder, i.e. a set of CP maps $\{\cD_m\}_{m\in [\mu]}$  
with a trace-preserving sum $\cD=\sum_m \cD_m$. The instrument has one quantum input and two outputs to detect the classical and quantum messages.
The encoding and decoding operations satisfy
\begin{align}\label{eq: DS F}
    F\left(\Phi_{\kappa}^{A'R}, (\cD\circ\cN^{\otimes n}\circ \cE_m) \Phi_{\kappa}^{A'R} \right)\geq 1-\epsilon, \quad \quad \text{and}\quad \quad 
    \Pr \left(\hat{M} \neq M \right)\leq \epsilon   \quad \forall m,
\end{align}
where $\hat{M}$ is the decoded classical message. 
The classical and quantum rates of this code are defined as $R_C=\frac{\log \mu}{n}$ and $R_Q=\frac{\log \kappa}{n}$, respectively.

\begin{figure}[t] 
  \includegraphics[width=0.5\textwidth]{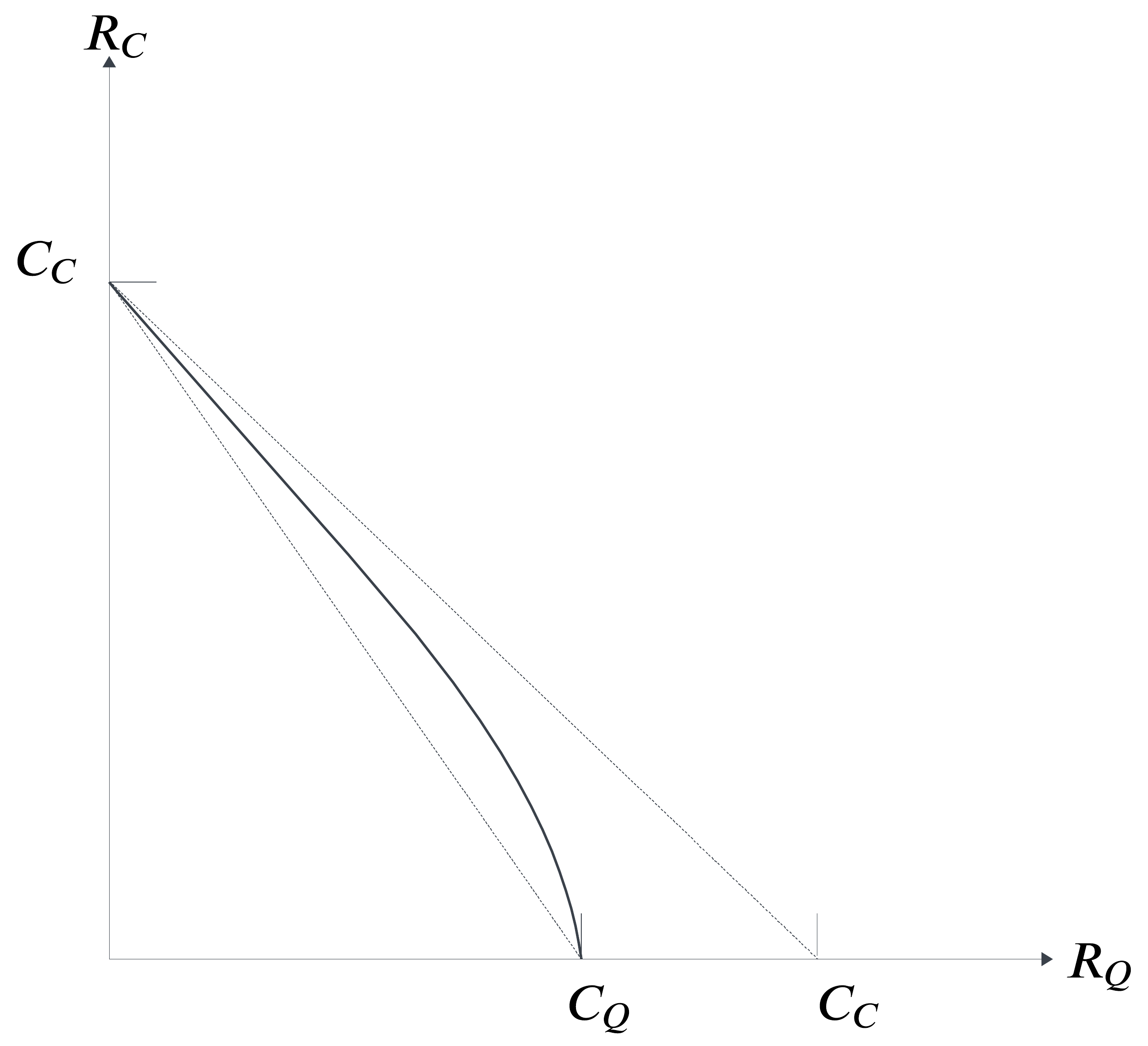} 
  \caption{The curve (solid line) is the trade-off curve of the simultaneous transmission of the  classical and quantum information over a quantum channel obtained by Devetak and Shor
  in \cite{Devetak-Shor-2005}. We call this curve the Devetak-Shor curve, and denote it by 
  $f_{\DS}(\cdot)$.
  The line $R_C=C_C-R_Q$ is an outer bound since the transmitted quantum entanglement may always be used to encode classical information at 1 cbit/qubit. The line $R_C=C_C-\frac{C_C}{C_Q}R_Q$ is an inner bound which can be achieved by time-sharing strategy.
  $C_C$ and $C_Q$ denote the classical capacity and the quantum capacity, respectively}
  \label{fig: DS curve}
\end{figure}

The trade-off  between the  classical and quantum rates
is stated in Theorem~\ref{Thm: Devetak-Shor} and  is illustrated in Fig.~\ref{fig: DS curve}. 
The trade-off curve is sandwiched between two lines $R_C=C_C-R_Q$ and $R_C=C_C-\frac{C_C}{C_Q}R_Q$
where $C_C$ and $C_Q$ are the classical and quantum capacity of the channel $\cN$, respectively. The former line is an outer bound since the transmitted quantum entanglement may always be used to encode classical information at 1 cbit/qubit.
The latter line is an inner bound which can be achieved by time-sharing strategy. 
We also remind that $C_Q \leq C_C$ holds for any quantum channel. 
%
%
%
\begin{theorem}\label{Thm: Devetak-Shor}\cite{Devetak-Shor-2005}
\noindent For the simultaneous transmission of classical and quantum information, defined as above, 
any rate pair $(R_Q,R_C)$ is asymptotically achievable if and only if 
\begin{align}
    R_C &\leq \lim_{l \to \infty} \frac{I(B^l : X)_{\sigma}}{l} \nonumber\\
    R_Q &\leq  \lim_{l \to \infty} \frac{I(R\>\rangle B^l X)_{\sigma}}{l}, \nonumber
\end{align}
and the above quantities are in terms of any state $\sigma$ of the following form 
\begin{align*} 
    \sigma^{B^lRX}&= (\cN^{\ox l}\ox \id_{RX})\varphi^{A^lRX} \quad \quad \text{and} \\
    \varphi^{A^lRX}&=\sum_{x} p(x)\proj{\varphi_{x}}^{A^lR}\ox \proj{x}^{X} , 
\end{align*}
where $\varphi_{x}^{A^l}$ is a state supported on $\cH_A^{\ox l}$ with a purification 
$\ket{\varphi_{x}}^{A^lR}$, and the dimension  of the system $X$ is bounded as $|X| \leq |\cH_A|^2+2$.
The capacity region is the union of all rate pairs $(R_Q,R_C)$ satisfying the above equalities. 
\end{theorem}

This paper is organized as follows. In the remainder of this section, we introduce some notation and convention that will be used throughout the paper. In Sec.~\ref{sec: transmission task}, we  define the quantum information transmission task and draw comparisons with the classical and quantum capacity problems. Sec.~\ref{sec: Main Results}
presents our main findings, establishing the generalized information as the capacity of the transmission task defined in the previous section.
In Sec.~\ref{sec: converse} and Sec.~\ref{sec: direct proof}, we prove the converse and achievability bounds, respectively. We further elaborate on our results in Sec.~\ref{sec: Discussion}. Finally, in the appendix we provide proofs and facts that are not incorporated within the main body of the paper.

\bigskip

\textbf{Notation.}
In this paper, quantum systems are associated with finite dimensional Hilbert spaces $A$, $R$, etc.,
whose dimensions are denoted by $|A|$, $|R|$, respectively. Since it is clear from the context, we slightly abuse the notation and let $Q$ denote both a quantum system and a quantum rate. 
The von Neumann entropy is defined as
\begin{align}
    S(\rho) := - \Tr\rho\log\rho \nonumber
\end{align}
Throughout this paper, $\log$ denotes by default the binary logarithm.
%
%
%
The conditional entropy and the conditional mutual information, $S(A|B)_{\rho}$ and $I(A:B|C)_{\rho}$, respectively, are defined in the same way as their classical counterparts: 
\begin{align*}
  S(A|B)_{\rho}   &= S(AB)_\rho-S(B)_{\rho}, \text{ and} \\ 
  I(A:B|C)_{\rho} &= S(A|C)_\rho-S(A|BC)_{\rho} \\
                    &= S(AC)_\rho+S(BC)_\rho-S(ABC)_\rho-S(C)_\rho.
\end{align*}
The fidelity between two states $\rho$ and $\xi$ is defined as 
\(
 F(\rho, \xi) = \|\sqrt{\rho}\sqrt{\xi}\|_1 
                = \Tr \sqrt{\rho^{\frac{1}{2}} \xi \rho^{\frac{1}{2}}},
\) 
with the trace norm $\|X\|_1 = \Tr|X| = \Tr\sqrt{X^\dagger X}$.
It relates to the trace distance by the Fuchs-van de Graaf inequality as follows \cite{Fuchs1999}:
\begin{equation}\label{eq: Fuchs-van de Graaf}
  1-F(\rho,\xi) \leq \frac12\|\rho-\xi\|_1 \leq \sqrt{1-F(\rho,\xi)^2}.
\end{equation}
When there are $n$ copies of a state, we use the following notation for systems and indices 
\begin{align*}
  x^n              &= x_1 x_2 \ldots x_n, \\
  \ket{x^n}        &= \ket{x_1} \ket{x_2} \cdots \ket{x_n}, \\
  p(x^n)           &= p(x_1) p(x_2)  \cdots p(x_n), \text{ and} \\
  \ket{\psi_{x^n}}^{A^n} &= \ket{\psi_{x_1}}^{A_1} \ket{\psi_{x_2}}^{A_2} \cdots \ket{\psi_{x_n}}^{A_n}.
\end{align*}

\section{The transmission task}\label{sec: transmission task}

In this section, we introduce a transmission task and call it the generalized capacity problem.
We assume that a source generates $m$ copies of a general mixed state ${\rho}^{A'R}$, and distributes them among an encoder, called Alice, and a reference system, who hold systems $A'$ and $R$, respectively. 
Alice applies an encoding operation, a CPTP map, $\cE_n:{A'}^m \to A^n$ to obtain a state
supported on ${\cH_A}^{\ox n}$, and then transmits this state to Bob via $n$ independent uses
of a channel $\cN:A \to B$. Bob receives a state supported on ${\cH_B}^{\ox n}$,
and applies a decoding operation $\cD_n:B^n \to {A'}^m$ to reconstruct system ${A'}^m$ while
preserving the correlations with system $R^m$. Our goal is to transmit maximal number of copies (m) of ${\rho}^{A'R}$ via $n$ uses of the channel $\cN$.
This task is depicted in Fig.~\ref{fig: diagram}.


The goal of the transmission task is well-defined and clear, however, we still need to
define a capacity quantity to establish connections with other capacity problems. One candidate
would be to define the capacity as $mS(A')/n$ similar to the original definition of the entanglement transmission capacity in \cite{fidelities-capacities-2000} defined as the entropy rate of the entanglement being transmitted over a channel.
The problem with this definition is that contrary to pure states, the entropy content 
of a mixed state can vary a lot depending on the size of the redundant system
in a Koashi-Imoto decomposition of the state \cite{KI2001,KI2002}.
Namely, according to the Koashi-Imoto decomposition of a state $\rho^{A'R}$, the correlations of system $A'$ with system $R$ can be categorized by decomposing system
$A'$, via an isometry, to three subsystems $C$, $N$ and $Q$ where subsystem
$C$ is only classically correlated with system $R$ and subsystem $Q$ is entangled 
with system $R$, however, system $N$ is a redundant system in the sense that
given the classical subsystem $C$, the reference $R$ is decoupled from subsystem $N$.
Therefore, for various mixed states correlated with system $R$ with same amount of classical and quantum correlations, the entropy content of the system $A'$ can vary a lot depending the entropy of the redundant subsystem $N$. 
This motivates us to define the generalized capacity as $mS(CQ)_{\omega}/n$, i.e. the entropy of the classical and quantum subsystems $C$ and $Q$ of the Koashi-Imoto decomposition.

\begin{figure}[t] 
  \includegraphics[width=0.92\textwidth]{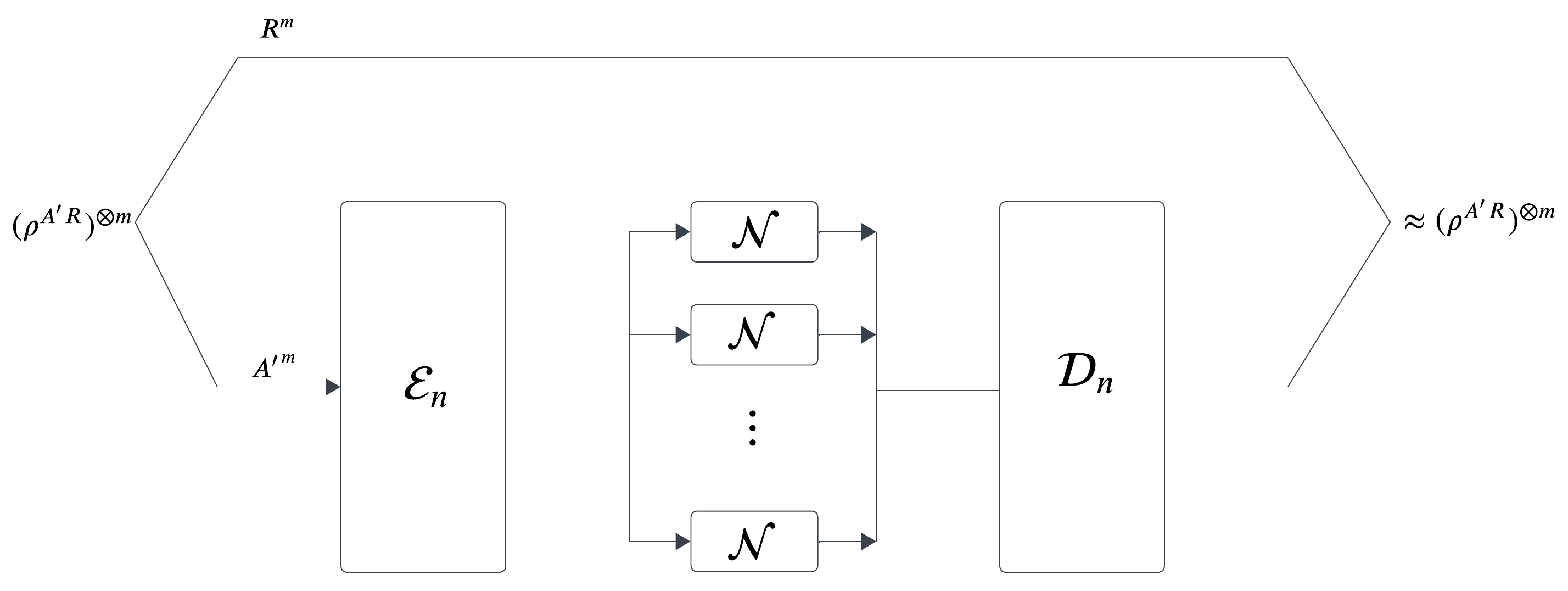} 
  \caption{The transmission task of the generalized capacity problem}
  \label{fig: diagram}
\end{figure}

We rigorously define the generalized capacity problem in Definition~\ref{def: capacity}.
Here we first review the Koashi-Imoto (KI) decomposition. 
This decomposition characterizes the correlations of system $A'$ with system $R$, and 
relates the nature of these correlations to the structure of CPTP maps which act on system $A'$ of a state $\rho^{A'R}$ and preserve the whole state, namely, a CPTP map $\Lambda:A'\to A'$ acting as $(\Lambda\ox \id_R)\rho^{A' R}=\rho^{A' R}$.
In summary, the KI decomposition states that there is a decomposition of the corresponding Hilbert space of $A'$ into three subspaces $C$, $N$ and $Q$, where a part of the state on subspace $C$ is classically correlated with $R$, a part of it on subspace $Q$
is entangled with the reference system $R$. The other part of the state on subspace $N$ is merely a redundant system. Then, any CPTP map $\Lambda$ with the aforementioned property acts non-trivially only on system $N$, and as an identity on the other systems. 

\begin{theorem}[\cite{KI2002,Hayden2004}]
\label{thm: KI decomposition}
Associated to the state $\rho^{AR}$, there are Hilbert spaces $C$, $N$ and $Q$
and an isometry $U_{\KI}:A \hookrightarrow C N Q$ such that:
\begin{enumerate}[(i)]
  \item The state $\rho^{AR}$ is transformed by $U_{\KI}$ as
    \begin{align}
        \label{eq:KI state}
      (U_{\KI}\!\otimes\! \1_R)\rho^{AR} (U_{\KI}^{\dagger}\!\otimes\! \1_R)
        \!\!= \!\!\sum_c p_c \proj{c}^{C}\!\! \otimes \mu_c^{N} \otimes \omega_c^{Q R} 
        =:\omega_{\rho}^{C N Q R},
    \end{align}
    where the set of vectors $\{ \ket{c}^{C}\}$ form an orthonormal basis for Hilbert space 
    $C$, and $p_c$ is a probability distribution over $c$. The states $\mu_c^{N}$ and 
    $\omega_c^{Q R}$ act  on the Hilbert spaces $N$ and $Q \otimes R$, respectively.

  \item For any CPTP map $\Lambda$ acting on system $A$ which leaves the state $\rho^{AR}$ 
    invariant, that is $(\Lambda \otimes \id_R )\rho^{AR}=\rho^{AR}$, every associated 
    isometric extension $U: A\hookrightarrow A E$ of $\Lambda$ with the environment system 
    $E$ is of the following form
    \begin{equation}
      U = (U_{\KI}\otimes \1_E)^{\dagger}
            \left( \sum_c \proj{c}^{C} \otimes U_c^{N} \otimes \1_c^{Q} \right) U_{\KI},
    \end{equation}
    where the isometries $U_c:N \hookrightarrow N E$ satisfy 
    $\Tr_E [U_c \omega_c U_c^{\dagger}]=\omega_c$ for all $c$.
    The isometry $U_{KI}$ is unique (up to a trivial change of basis of the Hilbert spaces 
    $C$, $N$ and $Q$). Henceforth, we call the isometry $U_{\KI}$ and the state 
    $\omega^{C N Q R}$ defined by Eq.~(\ref{eq:KI state})
    the Koashi-Imoto (KI) isometry and KI-decomposition of the state $\rho^{AR}$, respectively. 

  \item In the particular case of a tripartite system $CNQ$ and a state $\omega^{CNQR}$ already 
    in Koashi-Imoto form (\ref{eq:KI state}), property 2 says the following:
    For any CPTP map $\Lambda$ acting on systems $CNQ$ with 
    $(\Lambda \otimes \id_R )\omega^{CNQR}=\omega^{CNQR}$, every associated 
    isometric extension $U: CNQ\hookrightarrow CNQ E$ of $\Lambda$ with the environment system 
    $E$ is of the form
    \begin{equation}
      U = \sum_c \proj{c}^{C} \otimes U_c^{N} \otimes \1_c^{Q},
    \end{equation}
    where the isometries $U_c:N \hookrightarrow N E$ satisfy 
    $\Tr_E [U_c \omega_c U_c^{\dagger}]=\omega_c$ for all $c$.
\end{enumerate} 
\end{theorem}

\bigskip

\begin{definition}[The generalized capacity problem Fig.~\ref{fig: diagram}] \label{def: capacity}
For a quantum channel $\cN: A \to B$, a state $\rho^{A'R}$, $\epsilon > 0$ and $n\geq 1$, we define 
\begin{align}
    C_G(\rho^{A'R},n,\epsilon,\cN):=&  \max_{m} \frac{mS(CQ)_{\omega}}{n} \>\> \text{s.t.} 
 \nonumber\\ 
 & \exists \>\cE_n,\cD_n \quad 
      \text{with} \quad F\left((\rho^{A'R})^{\ox m},(\cD_n\circ\cN^{\otimes n}\circ \cE_n \otimes \id_{R^m} ) (\rho^{A'R})^{\ox m}\right)\geq 1-\epsilon, \nonumber
\end{align}
where $\rho^{A'R}$ is a general mixed quantum state with system $A'$ as the system to be transmitted through the channel and system $R$ as an inaccessible reference system. 
%
The maps $\cE_n: {A'}^m \to A^n$ and $\cD_n: B^n \to {A'}^{m}$ are respectively the encoding and decoding operations. We define the generalized capacity of a quantum channel $\cN$ 
for transmitting the state $\rho^{A'R}$ as
\begin{align}
 C_G(\rho^{A'R},\cN):=\inf_{\epsilon >0} \liminf_{n\to \infty} C_G(\rho^{A'R},n,\epsilon,\cN). \nonumber
\end{align}
\end{definition}


\begin{remark}\label{remark: special cases}
The classical capacity of a quantum channel \cite{Holevo1973} and the entanglement transmission capacity problem \cite{fidelities-capacities-2000,Devetak-capacity-2005} are special cases the generalized capacity problem of Definition~\ref{def: capacity}.
Namely, if we consider that the systems $A'$ and $R$ are both classical as $\rho^{A'R}=\sum_{l=1}^L \frac{1}{L}\proj{l}^{A'}\ox \proj{l}^R$, then $m$ copies of this state is $(\rho^{A'R})^{\ox m}=\sum_{l^m} \frac{1}{L^m}\proj{l^m}^{{A'}^m} \ox \proj{l^m}^{R^m}$.
In this case, preserving the fidelity of systems ${A'}^mR^m$  reduces to  preserving the average fidelity of system  ${A'}^m$, which follows from the definition of the fidelity.   
Hence, transmitting system ${A'}^m$ of this source is equivalent to transmitting 
a classical message from a set of $L^m$ number of classical messages over the channel. 
Moreover, if the shared state is pure with Schmidt decomposition $\ket{\rho}^{A'R}=\sum_{l=1}^L \frac{1}{\sqrt{L}} \ket{e_l}^{A'}\ox \ket{e_l}^R$, then 
$m$ copies of this state is $(\ket{\rho}^{A'R})^{\ox m}=\sum_{l^m=1}^{L^m}\frac{1}{\sqrt{L^m}}\ket{e_{l^m}}^{{A'}^m}\ox \ket{e_{l^m}}^{{R}^m}$.
Transmitting  ${A'}^m$ part of the state over $\cN^{\ox n}$ is 
equivalent to the task of transmitting entanglement over $\cN^{\ox n}$ (up to relabeling the local basis of ${A'}^m$ and $R^m$). 
%
%
\end{remark}

In order to make the problem technically easier to deal with, we introduce the notion of
equivalent sources and discuss that the task of transmitting $\rho^{A'R}$
is equivalent to the task transmitting $\omega_{\rho}^{CQR}$, i.e. the KI-decomposition of the state.
Namely, we say a given source $\omega^{KR}$ is equivalent to a source $\rho^{A'R}$ if there are
CPTP maps $\cT:A' \to K$ and $\cR:K \to A'$ in both directions
taking one to the other: 
\begin{align} \label{def: equivalent sources}
    \omega^{KR}=(\cT \otimes \id_R) \rho^{A'R} \text{ and } 
    \rho^{A'R}=(\cR \otimes \id_R) \omega^{KR}.
\end{align}
For equivalent sources the generalized capacities are equal, namely
$C_G(\rho^{A'R},n, \epsilon,\cN)=C_G(\omega^{KR},n, \epsilon,\cN)$
holds.
%
%
This follows because
for any code $(\cE_n,\cD_n)$ of block length $n$ and error $\epsilon$ for $(\rho^{A'R})^{\ox m}$, 
concatenating the encoding and decoding operations with $\cT$ and $\cR$, i.e. letting
$\cE'_n=\cE_n\circ\cR^{\otimes m}$ and $\cD'_n=\cT^{\otimes m}\circ\cD_n$, we get a code 
of the same error $\epsilon$ for $(\omega^{KR})^{\ox m}$. Analogously we can turn a code for 
$\omega^{CR}$ into a code for $\rho^{A'R}$.

%

The sources $\rho^{A'R} $ and $\omega_{\rho}^{CQR}$ are equivalent since the following CPTP maps relate these two state as
\begin{align} 
    \omega^{CQR}=(U_{\KI} \otimes \id_R) \rho^{A'R} \text{ and } 
    \rho^{A'R}=(\cR_{\KI} \otimes \id_R) \omega^{CR}, \nonumber
\end{align}
where $\cR_{\KI}:CNQ \to A'$ is the reverse-KI operation $\cR_{\KI}:CNQ \to A'$.

In the rest of the paper we consider the transmission of the source $(\omega_{\rho}^{CQR})^{\ox m}$ defined as
\begin{align}
    \omega_{\rho}^{C^mQ^nR^m}:=(\omega_{\rho}^{CQR})^{\ox m}=\sum_{c^m} p_{c^m} \proj{c^m}^{C^m}  \ox  \omega_{c^m}^{Q^m R^m }.
\end{align}
The encoder,  the channel and the decoder and their corresponding  Stinespring dilations are defined as follows
\begin{align}
& \cE_n: C^m Q^m \to A^n         &    & V_{\cE_n}: C^m Q^m \to A^n E_{\cE} \\
&\cN^{\ox n}: A^n \to B^n        & &   V_{\cN^{\ox n}}: A^n \to B^n E_{\cN} \\
&\cD_n: B^n \to \hat{C}^m \hat{Q}^m & & V_{\cD_n}: B^n \to \hat{C}^m \hat{Q}^m E_{\cD}.
\end{align}
Moreover, we consider the following  extension of the  source, which is useful for analyzing the transmission task. For a given index $c$, system ${R'}$ purifies the state on systems $Q$ and $R$, and $C'$ is a copy of the classical system $C$
\begin{align}
    \omega_{\rho}^{CQR{R'}{C'}}:=\sum_{c} p_{c} \proj{c }^{C }   \otimes \proj{\omega_{c }}^{QR{R'}} \otimes \proj{c}^{{C'}},
\end{align}
and $m$ copies of the above extended state is 
\begin{align}
    \omega_{\rho}^{C^mQ^mR^m{R'}^m{C'}^m}:=(\omega_{\rho}^{CQRR'C'})^{\ox m}=\sum_{c^m} p_{c^m} \proj{c^m}^{C^m}   \otimes \proj{\omega_{c^m}}^{Q^m R^m {R'}^m} \otimes \proj{c^m}^{{C'}^m}.
\end{align}
We define the output state of the encoding isometry, the channel isometry and the decoding isometry respectively as follows
\begin{align}
    \nu^{A^n E_{\cE}R^m{R'}^m{C'}^m}&:=(V_{\cE_n} \ox \id_{ R^m{R'}^m{C'}^m}) \omega_{\rho}^{C^mQ^mR^m{R'}^m{C'}^m} \nonumber\\
    &=\sum_{c^m} p_{c^m}   \proj{\nu_{c^m}}^{A^n E_{\cE} R^m{R'}^m} \ox \proj{c^m}^{{C'}^m},   \\
    \gamma^{B^n E_{\cE} E_{\cN}  R^m{R'}^m{C'}^m}&:= (V_{\cN^{\ox n}} \ox \id_{E_{\cE}R^m{R'}^m{C'}^m})\nu^{A^n E_{\cE}R^m{R'}^m{C'}^m} \nonumber\\
    &=\sum_{c^m} p_{c^m}  \proj{\gamma_{c^m}}^{B^n E_{\cE} E_{\cN}  R^m{R'}^m } \ox \proj{c^m}^{{C'}^m},   \\
    \xi^{\hat{C}^m  \hat{Q}^m E_{\cE} E_{\cN} E_{\cD}  R^m{R'}^m{C'}^m}&:=(V_{\cD_n} \ox \id_{E_{\cE} E_{\cN}  R^m{R'}^m{C'}^m})\sigma^{B^n E_{\cE} E_{\cN} R^m{R'}^m{C'}^m} \nonumber\\
&=\sum_{c^m} p_{c^m}   \proj{\xi_{c^m}}^{ \hat{C}^m  \hat{Q}^m E_{\cE} E_{\cN} E_{\cD}  R^m{R'}^m } \ox \proj{c^m}^{{C'}^m}.
\end{align}

\section{Main Results}\label{sec: Main Results}
%

We extend the coherent information to a new quantity and call it the generalized information. This quantity characterizes the generalized capacity of a channel as we show in Theorem~\ref{thm: main}.
\begin{definition}[Generalized information]\label{def: I_G}
We define the generalized information of a channel $\cN:A \to B$ for a state $\rho^{ARX}$ as
\begin{align}
    I_G(\rho^{ARX},\cN):=I(B:X)_{\sigma}+I(R\> \rangle BX)_{\sigma}, \nonumber
\end{align}
where $\rho^{ARX}=\sum_x p(x)\proj{\rho_x}^{AR}\ox \proj{x}^X$ is a cq state with the quantum systems $A$ and $R$, and the classical system $X$. The quantum mutual information and the coherent information are in terms of the state
$\sigma^{BRX}=\sum_x p(x)(\cN\ox\id_R)(\proj{\rho_x}^{AR})\ox \proj{x}^X$.
\end{definition}

\begin{remark}\label{remark T}
The generalized information reduces to the  mutual information or the coherent information if $R$ or $X$ is a trivial system, respectively. 
Namely, if system $R$ of the state $\rho^{ARX}$ is trivial, then 
\begin{align}
    I_G(\rho^{AX},\cN):=I(B:X)_{\sigma}, \nonumber
\end{align}
where $\sigma^{BX}=\sum_x p(x)\cN(\proj{\rho_x}^{A})\ox \proj{x}^X$.
Moreover, if system $X$ of the state $\rho^{ARX}$ is trivial, then 
\begin{align}
    I_G(\rho^{AR},\cN):=I(R\> \rangle B)_{\sigma}, \nonumber
\end{align}
where $\sigma^{BR}=(\cN\ox\id_R)\proj{\rho}^{AR}$.
\end{remark}
\medskip
Similar to the mutual information and the coherent information, the generalized information obeys data-processing inequality. Namely, for quantum channels $\cN_1: A \to B$ and $\cN_2: B \to C$
the inequality below holds
\begin{align}\label{eq: I_G data processing}
    I_G(\rho^{ARX},\cN_2 \circ \cN_1) \leq I_G(\rho^{ARX},\cN_1).  
\end{align}
This can be easily proven as below
\begin{align}
    I_G(\rho^{ARX},\cN_2 \circ \cN_1)&=I(C:X)_{\sigma'}+I(R\> \rangle CX)_{\sigma'} \nonumber\\
    &=I(C:X)_{\sigma'}+I(R:CX)_{\sigma'}-S(R)_{\sigma'} \nonumber\\
    &\leq I(B:X)_{\sigma}+I(R:BX)_{\sigma}-S(R)_{\sigma} \nonumber\\
    &=I_G(\rho^{ARX},\cN_1), \nonumber
\end{align}
where in the second and third lines the states are defined as ${\sigma'}^{CRX}=(\cN_2 \circ \cN_1 \ox\id_R)\rho^{ARX}$ and $\sigma^{BRX}=(\cN_1\ox\id_R)\rho^{ARX}$. In the third line, we apply data-processing inequality to the quantum mutual information.  

\medskip

\begin{figure}[t] 
\includegraphics[width=0.5\textwidth]{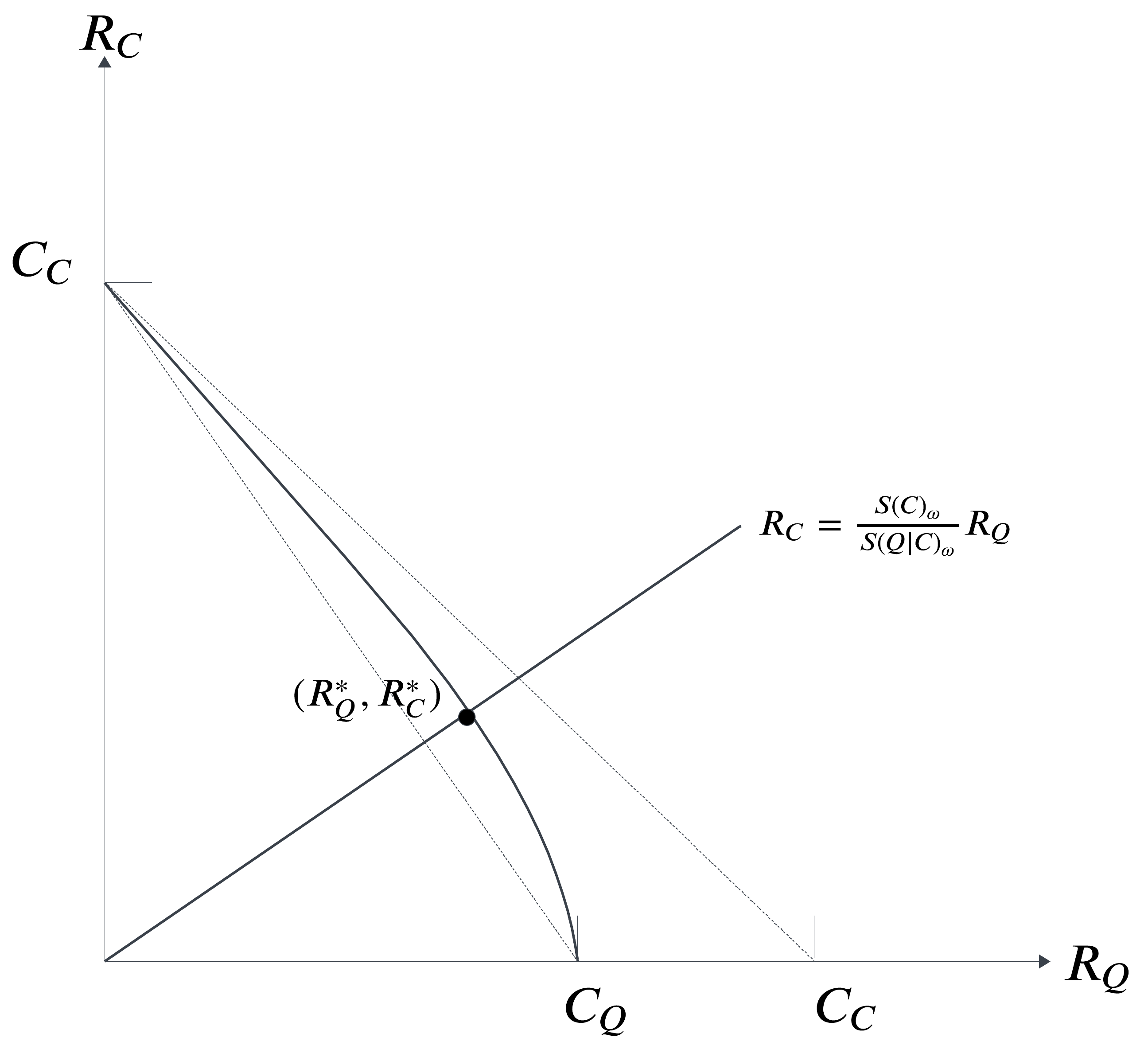}
\captionsetup{width=0.95\linewidth}%
\caption{ To find the generalized capacity of a channel $\cN$ for transmitting a state $\rho^{A'R}$, with KI decomposition $\omega_{\rho}^{CNQR}$, draw the line $R_C=\frac{S(C)_{\omega}}{S(Q|C)_{\omega}} R_Q$ and find the intersection of this line with the Devetak-Shor curve. Let $(R_Q^*,R_C^*)$ denote the intersected point. 
  Then, $C_G(\rho^{A'R},\cN)=R_Q^*+R_C^*$. This implies that the maximal number of copies of the state $\rho^{A'R}$, which can be transmitted per channel in the asymptotic limit, is  $\frac{R_Q^*+R_C^*}{S(CQ)_{\omega}}$. For a fully classical $\rho^{A'R}$ with $S(Q|C)_{\omega}=0$, the slop of the line is $\infty$, therefore, the generalized capacity is equal to the classical capacity as $C_G(\rho^{A'R},\cN)=C_C$.
  For a pure source $\ket{\rho}^{A'R}$  with $S(C)_{\omega}=0$,  the slop of the line is $0$, hence, the generalized capacity is equal to the quantum capacity as $C_G(\rho^{A'R},\cN)=C_Q$. The fact that the Devetak-Shor curve is below the line 
  $R_C=C_C-R_Q$ implies that the generalized capacity, which equals to the sum $R_Q^*+R_C^*$, is maximized for a fully classical state $\rho^{A'R}$}
\label{fig: Thm_diagram}
\end{figure}

\begin{theorem}\label{thm: main}
The generalized capacity of a quantum channel is characterised as follows 
\begin{align*}
    C_G(\rho^{A'R},\cN)=\lim_{l \to \infty} \frac{1}{l} \max_{\substack{\varphi^{A^lRX}: \\\frac{I(B^l: X)_{\sigma}}{I(R \>\rangle B^l X)_{\sigma}}=\frac{S(C)_{\omega}}{S(Q|C)_{\omega}} }}  I_G(\varphi^{A^lRX},\cN^{\ox l}),
\end{align*}
where the information quantities are with respect to the following states 
\begin{align}
  \varphi^{A^lRX}&=\sum_x \proj{\varphi_x}^{A^lR}\ox \proj{x}^X \nonumber   \\
  \sigma^{B^lRX}&=(\cN^{\ox l}\ox \id_{RX})\varphi^{A^lRX}. \nonumber
\end{align}
Geometrically, the generalized capacity is characterized  by the intersection point of the line $R_C=\frac{S(C)_{\omega}}{S(Q|C)_{\omega}} R_Q$ and the Devetak-Shore curve for the simultaneous transmission of classical and quantum information obtained  in \cite{Devetak-Shor-2005}.
Namely, let $(R_Q^*,R_C^*)$ be the intersected point, then 
\begin{align*}
C_G(\rho^{A'R},\cN)=R_Q^*+R_C^*.
\end{align*}
This is illustrated in Fig.~\ref{fig: Thm_diagram}. 
The dimension  of the system $X$ is bounded as $|X| \leq |\cH_{A}|^2+2$.
\end{theorem}

\noindent Above, the cardinality bound on system $X$ directly follows from Theorem~\ref{Thm: Devetak-Shor} since the generalized information is the sum of the rate terms in Theorem~\ref{Thm: Devetak-Shor} for which the cardinality bound on system $X$ holds.

\medskip
Remark~\ref{remark T} and Theorem~\ref{thm: main} imply the classical capacity and 
the entanglement transmission capacity theorems as follows.
\begin{corollary}
The classical capacity $C_C(\cN)$ and the entanglement transmission capacity (the quantum capacity) 
 $C_Q(\cN)$ are special cases of the generalized capacity for the transmission of a fully
 classical state and a pure entangled state, respectively as follows
\begin{align*}
    C_C(\cN)&=C_G(\rho^{A'R},\cN)\\
    C_Q(\cN)&=C_G(\ket{\psi}^{A'R},\cN),
\end{align*} 
where systems $A'$ and $R$ of the state $\rho^{A'R}=\sum_{x}p(x)\proj{x}^{A'}\ox \proj{x}^R$ are classical.
\end{corollary}

\section{Converse Proof}\label{sec: converse}
In this section, we prove the converse bound (an upper bound) for the rate of Theorem~\ref{thm: main}. We first define two functions in the following definition which appear
in the converse bounds. Then, in Lemma~\ref{lemma:Y_epsilon properties}, we find 
some properties of these functions, which we use to analyze the converse bounds. We leave the proof of this lemma to Sec.~\ref{sec: Proof of Lemma Y W epsilon} in the appendix.

\begin{definition}
  \label{def:Y_epsilon}
  For the KI decomposition (without the redundant system $N$) 
  $\omega^{C Q R}=\sum_{c} p_c \proj{c}^{C} \otimes \omega_{c}^{Q R}$
  of the state $\rho^{A'R}$ and $\epsilon \geq 0$, define
  \begin{align*}
    Y_\epsilon(\omega) &:=  
        \max S(\hat{Q} RR'|\hat{C})_\tau 
                  &&\text{ s.t. } U:C Q \rightarrow \hat{C} \hat{Q} E
                  \text{ is an isometry with } 
                  F( \omega^{C Q R},\tau^{\hat{C} \hat{Q}R})  \geq 1- \epsilon,\\
    W_\epsilon(\omega) &:=  
        \max S(\hat{C} |C')_\tau 
                  &&\text{ s.t. } U:C Q \rightarrow \hat{C} \hat{Q} E
                  \text{ is an isometry with } 
                  F( \omega^{C Q R},\tau^{\hat{C} \hat{Q}R})  \geq 1- \epsilon,              
\end{align*}
where the conditional entropy is with respect to the state $\tau$ defined as
\begin{align*}
  \omega^{CQR{R'}{C'}}&:=\sum_{c} p_{c} \proj{c }^{C }   \otimes \proj{\omega_{c }}^{QR{R'}} \otimes \proj{c}^{{C'}},    \\
  \tau^{\hat{C}  \hat{Q} ERR' C'}
     &:= (U \otimes \1_{RR'C'}) \omega^{CQ R R'C'}    (U^{\dagger} \otimes \1_{RR'C'}), \\
     &=\sum_{c} p_{c} \proj{\tau_{c}}^{\hat{C}\hat{Q}ER{R'}} \otimes \proj{c}^{{C'}},
\end{align*}
where $\omega^{CQR{R'}{C'}}$ is an extension  of the state $\omega^{CQR}$.
\end{definition}
In this definition, the dimension of the environment is w.l.o.g. bounded as $|E| \leq (|C||Q|)^2$ because the input and output dimensions of the channel are fixed as $|C||Q|$; hence, the optimisation is of a continuous function over a compact domain, so we have a maximum rather than a supremum.
\begin{lemma}
  \label{lemma:Y_epsilon properties}
  The functions $Y_\epsilon(\omega)$ and $W_\epsilon(\omega)$ have the following properties:
  \begin{enumerate}
    \item They are  non-decreasing functions of $\epsilon$. 
    \item They are concave in $\epsilon$.
    \item They are continuous for $\epsilon \geq 0$. 
    \item They are sub-additive, i.e. for any two states $\omega_1^{C_1 Q_1 R_1}$ and $\omega_2^{C_2 Q_2 R_2}$ and for $\epsilon \geq 0$,
    \begin{align*}
        &Y_{\epsilon}(\omega_1 \otimes \omega_2) \leq  Y_{\epsilon}(\omega_1) +Y_{\epsilon}(\omega_2),\\
        &W_{\epsilon}(\omega_1 \otimes \omega_2) \leq  W_{\epsilon}(\omega_1) +W_{\epsilon}(\omega_2).
    \end{align*}
    \item At $\epsilon=0$, $Y_0(\omega) =0$ and $W_0(\omega) =0$ hold.
  \end{enumerate}
\end{lemma}

%

We are now well equipped to bound the capacity term $mS(CQ)_{\omega}/n$. We first obtain the following chain of inequalities
\begin{align}\label{eq: converse 1}
    mS(Q|{C})_{\omega}&=S(Q^m|{C}^m)_{\omega} \nonumber\\
    &\leq  S(\hat{Q}^m |\hat{C}^m)_{\xi}+m\delta(m,\epsilon) \nonumber\\
    &=S(\hat{Q}^m R^m {R'}^m|\hat{C}^m)_{\xi}-S( R^m {R'}^m| \hat{C}^m  \hat{Q}^m)_{\xi}+m\delta(m,\epsilon) \nonumber \\
    &=S(\hat{Q}^m R^m {R'}^m|\hat{C}^m)_{\xi}+I( R^m{R'}^m: \hat{C}^m  \hat{Q}^m)_{\xi}-S( R^m {R'}^m)_{\xi}+m\delta(m,\epsilon) \nonumber \\
    &\leq S(\hat{Q}^m R^m {R'}^m|\hat{C}^m)_{\xi}+I( R^m{R'}^m: B^n)_{\gamma}-S( R^m {R'}^m)_{\xi}+m\delta(m,\epsilon) \nonumber \\
    &\leq S(\hat{Q}^m R^m {R'}^m|\hat{C}^m)_{\xi}+I( R^m{R'}^m: B^n {C'}^m)_{\gamma}-S( R^m {R'}^m)_{\xi}+m\delta(m,\epsilon) \nonumber \\
    &=S(\hat{Q}^m R^m {R'}^m|\hat{C}^m)_{\xi}-S( R^m{R'}^m| B^n {C'}^m)_{\gamma}+m\delta(m,\epsilon) \nonumber \\
    &\leq Y_{\epsilon}(\omega^{\ox m}) -S( R^m{R'}^m| B^n {C'}^m)_{\gamma}+m\delta(m,\epsilon) \nonumber \\
    &\leq mY_{\epsilon}(\omega) -S( R^m{R'}^m| B^n {C'}^m)_{\gamma}+m\delta(m,\epsilon) \nonumber \\
    &=mY_{\epsilon}(\omega)+\sum_{c^m} p_{c^m}I_c(\omega_{c^m}^{Q^m},\cN^{\ox n}\circ \cE_n )+m\delta(m,\epsilon),  
\end{align}
where the second line is due to the decodability of the information:
the 
output state on systems $\hat{C}^m \hat{Q}^m$ is $2\sqrt{2\epsilon}$-close 
to the original state $C^m  Q^m$ in trace norm; then the inequality follows 
by applying the Fannes-Audenaert inequality 
\cite{Fannes1973,Audenaert2007}, where 
$\delta(m,\epsilon)=\sqrt{2\epsilon} \log(|C||Q|) + \frac1m h(\sqrt{2\epsilon})$.
The fifth and sixth lines are due to data-processing inequality. 
The eighth line follows from Definition~\ref{def:Y_epsilon}.
The penultimate line
is due to Lemma~\ref{lemma:Y_epsilon properties} point 4.
The last line follows from the definition of the coherent information for the 
channel $\cN^{\ox n}\circ \cE_n$. Moreover, we obtain the following bound
on the classical part of the state $\omega$
\begin{align}\label{eq: converse 2}
    mS({C})_{\omega}&=S({C}^m)_{\omega} \nonumber\\
    & \leq S(\hat{C}^m)_{\omega} +m\delta(m,\epsilon)\nonumber\\
    & = S(\hat{C}^m|{C'}^m)_{\omega}+I(\hat{C}^m:{C'}^m)_{\omega}+m\delta(m,\epsilon) \nonumber\\
    &\leq S(\hat{C}^m|{C'}^m)_{\omega}+I(B^n:{C'}^m)_{\gamma}+m\delta(m,\epsilon)\nonumber\\
    &\leq W_{\epsilon}(\omega^{\ox m})+I(B^n:{C'}^m)_{\gamma}+m\delta(m,\epsilon)\nonumber\\
    &\leq mW_{\epsilon}(\omega)+I(B^n:{C'}^m)_{\gamma}+m\delta(m,\epsilon)\nonumber\\
    &= mW_{\epsilon}(\omega)+ \chi(\omega^{C^mQ^m},\cN^{\ox n}\circ \cE_n)+m\delta(m,\epsilon),  
\end{align}
where the second line is due to the decodability of the information:
the 
output state on systems $\hat{C}^m \hat{Q}^m$ is $2\sqrt{2\epsilon}$-close 
to the original state $C^m  Q^m$ in trace norm; then the inequality follows 
by applying the Fannes-Audenaert inequality 
\cite{Fannes1973,Audenaert2007}, where 
$\delta(m,\epsilon)=\sqrt{2\epsilon} \log(|C||Q|) + \frac1m h(\sqrt{2\epsilon})$.
The forth line is due to data-processing inequality.
The fifth line is due to Definition~\ref{def:Y_epsilon}.
The penultimate line follows from Lemma~\ref{lemma:Y_epsilon properties} point 4.
The last line follows from the definition of the Holevo information for the 
channel $\cN^{\ox n}\circ \cE_n$.

To obtain the converse bound in the final step at the end of this section, we will divide Eq.~(\ref{eq: converse 1}) and Eq.~(\ref{eq: converse 2}) by $n$ and then take  the limit of $n \to \infty$ and $\epsilon \to 0$. Note that the ratio $m/n$ is always bounded in that the number of copies of any source, which can be transmitted per channel use, is bounded. This implies that $\lim_{\epsilon \to 0}m\delta(m,\epsilon)/n=0$ and $\lim_{\epsilon \to 0}mY_{\epsilon}(\omega)/n=\lim_{\epsilon \to 0}mW_{\epsilon}(\omega)/n=0$, due to Lemma~\ref{lemma:Y_epsilon properties} points 3 and 5. Therefore, 
to make the expressions neater, we omit the terms $\delta(m,\epsilon)$, $Y_{\epsilon}(\omega)$ and $W_{\epsilon}(\omega)$ in the rest of the calculations.   

In both Eq.~(\ref{eq: converse 1}) and Eq.~(\ref{eq: converse 2}) we obtain 
bounds on the number of copies $m$. Therefore, $m$ is bounded by the minimum of these two bounds as
\begin{align}\label{eq: m bound}
    m &\leq \min \left\{  \frac{\sum_{c^m} p_{c^m}I_c(\omega_{c^m}^{Q^m},\cN^{\ox n}\circ \cE_n )}{S(Q|C)_{\omega}},\frac{\chi(\omega^{C^mQ^m},\cN^{\ox n}\circ \cE_n )}{S(C)_{\omega}}  \right\} \nonumber \\
    & \leq \max_{\Omega^{C^mQ^mR^m{C'}^m},\cE_n} \min \left\{  \frac{\sum_{c^m} q_{c^m}I_c(\Omega_{c^m}^{Q^m},\cN^{\ox n}\circ \cE_n )}{S(Q|C)_{\omega}},\frac{\chi(\Omega^{C^mQ^m},\cN^{\ox n}\circ \cE_n )}{S(C)_{\omega}}  \right\}   \\
    & = \max_{\Omega^{C^mQ^mR^m{C'}^m},\cE_n} \min \left\{  \frac{I(R^m\>\rangle B^n {C'}^m)_{\Theta}}{S(Q|C)_{\omega}},\frac{I(B^n:{C'}^m)_{\Theta}}{S(C)_{\omega}}  \right\}  \nonumber
\end{align}
where in the second inequality we take the maximum over the encoders $\cE_n$ and all states of the form 
$\Omega^{C^mQ^mR^m{C'}^m}$ defined as follows, and the coherent information and the mutual information are with respect to the state $\Theta$ defined below
\begin{align}
    \Omega^{C^mQ^nR^m{C'}^m}&=(\Omega^{CQR})^{\ox m}=\sum_{c^m} q_{c^m} \proj{c^m}^{C^m}  \ox  \proj{\Omega_{c^m}}^{Q^m R^m } \ox \proj{c^m}^{{C'}^m}, \nonumber\\
    \Xi^{A^nR^m{C'}^m}&=(\cE_n \ox \id_{R^m{C'}^m}) \Omega^{C^mQ^nR^m{C'}^m}  
    =\sum_{c^m} q_{c^m}     \Xi_{c^m}^{A^m  R^m } \ox \proj{c^m}^{{C'}^m} \nonumber\\
    \Theta^{B^nE_{\cN}R^m{C'}^m}&=(V_{\cN^{\ox n}} \ox \1_{R^m{C'}^m})\Xi^{A^nR^m{C'}^m} (V_{\cN^{\ox n}} \ox \1_{R^m{C'}^m})^{\dagger} 
    =\sum_{c^m} q_{c^m}     \Theta_{c^m}^{B^m E_{\cN} R^m } \ox \proj{c^m}^{{C'}^m},\nonumber
\end{align} 
where $V_{\cN^{\ox n}}$ denotes the Stinespring isometry of the channel $\cN^{\ox n}$.
Now, we show how we can get rid of the encoding map $\cE_n$ in the maximization of Eq.~(\ref{eq: m bound}).
To this end, we define extensions of the states  $\Xi$ and $\Theta$ by
considering the spectral decomposition of the state $ \Xi_{c^m}^{A^m  R^m }=\sum_j q_{j|c^m} \proj{\Xi_{c^mj}}^{A^n R^m }$ and introducing a new register $J$, which stores the indices of the eigenvalues of this state as the following
\begin{align}\label{eq: big Xi}
    \Xi^{A^nR^m{C'}^mJ} 
    =\sum_{c^m,j} q_{c^m} q_{j|c^m} \proj{\Xi_{c^mj}}^{A^n R^m } \ox \proj{c^m}^{{C'}^m} \ox \proj{j}^J, 
\end{align}
and we define the extended state $\Theta$ as 
\begin{align}
    \Theta^{B^nE_{\cN}R^m{C'}^mJ}&=\sum_{c^m,j} q_{c^m} q_{j|c^m}(V_{\cN^{\ox n}} \ox \1_{R^m{C'}^m}) \proj{\Xi_{c^mj}}^{A^n R^m } (V_{\cN^{\ox n}} \ox \1_{R^m{C'}^m})^{\dagger} \ox \proj{c^m}^{{C'}^m} \ox \proj{j}^J \nonumber \nonumber\\
    &=\sum_{c^m,j} q_{c^m} q_{j|c^m} \proj{\Theta_{c^mj}}^{B^m E_{\cN} R^m } \ox \proj{c^m}^{{C'}^m} \ox \proj{j}^J. \nonumber
\end{align}
For the above extended state, the following holds due to the strong sub-additivity of the entropy
\begin{align}
    I(R^m\>\rangle B^n {C'}^m)_{\Theta}&\leq I(R^m\>\rangle B^n {C'}^m J)_{\Theta} \nonumber \\
    I(B^m:{C'}^m)_{\Theta}&\leq I(B^m:{C'}^mJ)_{\Theta}. \nonumber 
\end{align}
Therefore, as we are interested in maximizing the mutual information and the coherent information 
of the channel, we can assume w.l.o.g. that the input state of the channel is a state
of the form of Eq.~(\ref{eq: big Xi}), which is cq with classical registers ${C'}^m J$ and 
pure quantum states of the ensemble $\ket{\Xi_{c^mj}}^{A^n R^m }$. More generally, 
we can assume that the input and output states of the channel is of the following form
\begin{align}\label{eq: varphi and sigma}
  \varphi^{A^nRX}&=\sum_x p_x \proj{\varphi_x}^{A^nR}\ox \proj{x}^X \nonumber   \\
  \sigma^{B^nRX}&=(\cN^{\ox n}\ox \id_{RX})\varphi^{A^nRX}. 
\end{align}
Therefore, we can rewrite Eq.~(\ref{eq: m bound}) as the following
\begin{align}
    m &\leq \max_{\Omega^{C^mQ^mR^m{C'}^m},\cE_n} \min \left\{  \frac{I(R^m\>\rangle B^n {C'}^m)_{\Theta}}{S(Q|C)_{\omega}},\frac{I(B^n:{C'}^m)_{\Theta}}{S(C)_{\omega}}  \right\}  \nonumber\\
    &\leq \max_{\Xi^{A^nR^m{C'}^mJ}} \min \left\{  \frac{I(R^m\>\rangle B^n {C'}^mJ)_{\Theta}}{S(Q|C)_{\omega}},\frac{I(B^n:{C'}^mJ)_{\Theta}}{S(C)_{\omega}}  \right\}  \nonumber\\
    &\leq  \max_{\varphi^{A^nRX}} \min \left\{  \frac{I(R \>\rangle B^n X)_{\sigma} }{S(Q|C)_{\omega}},\frac{I(B^n: X)_{\sigma} }{S(C)_{\omega}}  \right\}.  \nonumber
\end{align}
We divide the above equation by $n$ and take the limit of $n \to \infty$.
Then, we can write the arguments of the above minimization in terms of the Devetak-Shor curve
$(R_Q,f_{\DS}(R_Q))$
\begin{align}
    \lim_{n \to \infty}\frac{m}{n} \leq  \max_{R_Q} \min \left\{  \frac{R_Q}{S(Q|C)_{\omega}},\frac{f_{\DS}(R_Q)}{S(C)_{\omega}}  \right\}.  \nonumber
\end{align}
The maximum above is achieved for a value of $R_Q$ satisfying $\frac{R_Q} {S(Q|C)_{\omega}}=\frac{f_{\DS}(R_Q)}{S(C)_{\omega}}$.
Hence, considering this equality we can bound $mS(CQ)_{\omega}$ as
\begin{align}
    \lim_{n \to \infty}\frac{mS(CQ)_{\omega}}{n}&=\lim_{n \to \infty}\frac{mS(Q|C)_{\omega}+mS(C)_{\omega}}{n} \nonumber\\
    &\leq S(Q|C)_{\omega}  \frac{R_Q}{S(Q|C)_{\omega}}+S(C)_{\omega}\frac{f_{\DS}(R_Q)}{S(C)_{\omega}}  \nonumber\\
    &=R_Q+ f_{\DS}(R_Q). \nonumber
\end{align}
We can rewrite the converse bound as follows
\begin{align}
    \lim_{n \to \infty} \frac{mS(CQ)_{\omega}}{n}&\leq  \lim_{n \to \infty} \frac{1}{n}
    \max_{\substack{\varphi^{A^nRX}: \\\frac{I(B^n: X)_{\sigma}}{I(R \>\rangle B^n X)_{\sigma}}=\frac{S(C)_{\omega}}{S(Q|C)_{\omega}} }}  I_G(\varphi^{A^nRX},\cN^{\ox n}) \nonumber\\
    &= \lim_{n \to \infty} \frac{1}{n} \max_{\substack{\varphi^{A^nRX}: \\\frac{I(B^n: X)_{\sigma}}{I(R \>\rangle B^n X)_{\sigma}} 
    =\frac{S(C)_{\omega}}{S(Q|C)_{\omega}} }}  I(B^n: X)_{\sigma}+ I(R \>\rangle B^n X)_{\sigma}. \nonumber
\end{align}

\section{Achievability Proof}\label{sec: direct proof}
 
\noindent We show that the following rate is achievable for any state $\rho^{ARX}=\sum_x p(x)\proj{\rho_x}^{AR}\ox \proj{x}^X$ 
\begin{align}
    I_G(\rho^{ARX},\cN)=I(B:X)_{\sigma}+I(R\> \rangle BX)_{\sigma}, \nonumber
\end{align}
where   
$\sigma^{BRX}=\sum_x p(x)(\cN\ox\id_R)(\proj{\rho_x}^{AR})\ox \proj{x}^X$, and it
satisfies $\frac{I(B:X)_{\sigma}}{I(R\> \rangle BX)_{\sigma}}=\frac{S(C)_{\omega}}{S(Q|C)_{\omega}}$.
In Sec.~\ref{sec: single-letter curve} of the appendix, we further explain why we can always find such a state $\sigma$.
The regularized formula is then obtained by additional blocking.
Namely, we show that for any $\epsilon,\delta >0$ and sufficiently large $n$, there is an
$(n,\epsilon)$ code of rate $I_G(\rho^{ARX},\cN)-\delta$.
First the encoder, Alice, reduces the dimension of her systems $C^m Q^m$ by projecting them onto the typical subspace to obtain the following normalized state
\begin{align}\label{eq: normalized typical omega}
   \sum_{c^m \in \cT^m_{\delta,p}} \frac{p_{c^m}}{(1-\epsilon)} \proj{c^m}^{C^m}   \otimes \proj{\omega'_{c^m}}^{Q^m R^m {R'}^m} \otimes \proj{c^m}^{{C'}^m},   
\end{align}
where any for $c^m \in \cT^m_{\delta,p}$ and the conditional typical projector $\Pi_{\delta,\omega_{c^n}}^n$,  the normalized state $\omega'_{c^m}$ is defined as
 \begin{align}
  \proj{\omega'_{c^m}}^{Q^m R^m {R'}^m}=\frac{(\Pi_{\delta,\omega_{c^n}}^n\ox \1_{R^m {R'}^m})  \proj{\omega_{c^m}}^{Q^m R^m {R'}^m}(\Pi_{\delta,\omega_{c^n}}^n\ox \1_{R^m {R'}^m})^{\dagger}}{\Tr \left((\Pi_{\delta,\omega_{c^n}}^n\ox \1_{R^m {R'}^m})  \proj{\omega_{c^m}}^{Q^m R^m {R'}^m}(\Pi_{\delta,\omega_{c^n}}^n\ox \1_{R^m {R'}^m})^{\dagger}\right)}.  \nonumber
\end{align}
All the above notions and properties such as the typical subspace and its dimension
is explained in Sec.~\ref{sec: Miscellaneous definitions and facts}.
Alice uses the protocol of Lemma~\ref{lemma: DS extension} to transmit systems $C^m$ and $Q^m$. Note that the state in Eq.~(\ref{eq: normalized typical omega}) is exactly of the form of the state of Eq.~(\ref{eq: DS extended state}) considered in Lemma~\ref{lemma: DS extension} with $A'=Q^m$, $R=R^m{R'}^m$, $M=C^m$, and $M'={C'}^m$.
The dimensions of the classical system and the quantum system are bounded as $2^{m[S(C)_{\omega}+c\delta]}$ and $2^{m[S(Q|C)_{\omega}+c\delta]}$, respectively.   
We let Alice choose $m$ such that $2^{m[S(Q|C)_{\omega}+c\delta]}=2^{nI(R \> \rangle BX)_{\sigma}}$
and $2^{m[S(C)_{\omega}+c\delta]}=2^{nI(B:X)_{\sigma}}$. This implies that 
$mS(CQ)_{\omega}=n[I(B:X)_{\sigma}+I(R \> \rangle BX)_{\sigma}]-2m\delta$. 
By choosing these rates, Lemma~\ref{lemma: DS extension} ensures that systems $C^m$ and $Q^m$ are reconstructed at the decoder, and the fidelity on the systems $C^mQ^m R^m {R'}^m {C'}^m$ is preserved above $1-\epsilon$. This implies high fidelity on the reduced state on systems $C^mQ^m R^m$, due to monotonicity of the fidelity under partial trace.  

\bigskip

In Lemma~\ref{lemma: DS extension}, we obtain a trade-off between the classical and quantum rates of the following task. Consider the transmission of systems $A'$ and $M$ of the following source  
\begin{align}\label{eq: DS extended state}
        \omega^{A'RMM'}= \sum_{m=1}^{\mu} \alpha_{m} \proj{\phi_{m}}^{A'R}\ox\proj{m}^M
      \ox  \proj{m}^{M'},   
\end{align}
where for all $m$, $\ket{\phi_m}^{A'R}=\sum_{k=1}^{\kappa} \sqrt{\lambda_k} \ket{v_k}^{A'} \ket{w_k}^{R}$ is any pure state shared between $A'$ and $R$ with Schmidt rank $\kappa$ ($\sum_k \lambda_k=1$ and $0<\lambda_k \leq 1$ for all $k$), and $\{\alpha_m\}_{m=1}^{\mu}$ is a probability distribution. The systems $M$ and $M'$ are classical, and $R$ and $M'$ are reference systems.  
Alice applies an encoding operation $\cE_n:{A'}M \to A^n$ to obtain a state
supported on ${\cH_A}^{\ox n}$, and then transmits this state to Bob 
via $n$ independent uses of a quantum channel $\cN: A \to B$.
Bob receives a state supported on ${\cH_B}^{\ox n}$,
and applies a decoding operation $\cD_n:B^n \to {A'} {M}$ to reconstruct  the original source with the following fidelity
\begin{align}\label{eq: extended DS F}
    F\left( \omega^{A'RMM'} , (\cD_n\circ\cN^{\otimes n}\circ \cE_n \otimes \id_{RM'} ) \omega^{A'RMM'} \right)\geq 1-\epsilon.  
\end{align}
The classical and quantum rates of this code are defined as $R_C=\frac{\log \mu}{n}$ and $R_Q=\frac{\log \kappa}{n}$, respectively.
\begin{lemma}\label{lemma: DS extension}
%
%
%
%
For the above task, any rate pair $(R_Q=\frac{\log \kappa}{n},R_C=\frac{\log \mu}{n})$ is asymptotically achievable if 
\begin{align}
    R_C &\leq \lim_{l \to \infty} \frac{I(B^l : X)_{\sigma}}{l} \nonumber\\
    R_Q &\leq  \lim_{l \to \infty} \frac{I(R\>\rangle B^l X)_{\sigma}}{l}, \nonumber
\end{align}
where the above quantities are in terms of any state $\sigma$ of the following form 
\begin{align*} 
    \sigma^{B^lRX}&= (\cN^{\ox l}\ox \id_{RX})\varphi^{A^lRX} \quad \quad \text{and} \\
    \varphi^{A^lRX}&=\sum_{x} p(x)\proj{\varphi_{x}}^{A^lR}\ox \proj{x}^{X} , 
\end{align*}
where $\varphi_{x}^{A^l}$ is a state supported on $\cH_A^{\ox l}$ with a purification 
$\ket{\varphi_{x}}^{A^lR}$, and the dimension  of the system $X$ is bounded as $|X| \leq |\cH_A|^2+2$.
\end{lemma}
\begin{proof}
Let's first a simpler task of 
transmitting 
systems $A'$ and $M$ of the following state 
\begin{align}
        \rho^{A'RMM'}=\proj{\Phi_{\kappa}}^{A'R}\ox \sum_{m=1}^{\mu}\frac{1}{\mu}\proj{m}^M
      \ox  \proj{m}^{M'},   \nonumber
\end{align}
$\ket{\Phi_{\kappa}}^{A'R}$ is a maximally entangled state of dimension $\kappa$. 
The systems $M$ and $M'$ are classical, and $R$ and $M'$ are reference systems. 
An $(n,\epsilon)$ code for this task consists of an encoder $\cE_n: A' M \to A^n$ and a decoder  $\cD_n: B^n \to A' M$ such that the  fidelity is preserved as the following
\begin{align}\label{eq: extended 2 DS F}
    F\left( \rho^{A'RMM'} , (\cD_n\circ\cN^{\otimes n}\circ \cE_n \otimes \id_{RM'} ) \rho^{A'RMM'} \right)\geq 1-\epsilon,  
\end{align}
and the encoded system $A^n$ is transmitted to the decoder via the channel $\cN^{\ox n}$.
%

This is the same task defined by Devetak and Shor except that the above fidelity is slightly different from Eq.~(\ref{eq: DS F}), where the fidelity on the reduced systems $A'R$ and $MM'$ is considered. 
However, we show that for very small $\epsilon$, these two fidelity criteria
are equivalent. One direction of the equivalence 
follows from  monotonicity of the fidelity under partial trace, namely,  the fidelity of Eq.~(\ref{eq: extended 2 DS F}) implies that the fidelity of Eq.~(\ref{eq: DS F}) on the reduced systems $A'R$ and $MM'$ is preserved. The other direction of the equivalence follows from Lemma~\ref{lemma: almost pure reduced state}. Namely, since the decoded system on $A'R$ is expected to be an almost pure state,  high fidelity of Eq.~(\ref{eq: DS F}) on the reduced systems $A'R$ and $MM'$ implies high fidelity on the global state of Eq.~(\ref{eq: extended 2 DS F}).

Now consider transmitting systems $A'$ and $M$ of the source defined in Eq.~(\ref{eq: DS extended state}).
%
Depending on the message $m$, instead of implementing a quantum code for transmitting half of
a maximally entangled state, as explained in the proof of Theorem~1 in \cite{Devetak-Shor-2005}, we use a quantum code for transmitting system $A'$ of the state $\ket{\phi_{m}}^{A'R}$. 
This does not change the ``type'' of the quantum code considered in \cite{Devetak-Shor-2005} in that
for any pure state on systems $A'R$, a quantum code with a given type can be constructed
 \cite{Devetak-capacity-2005}.
Hence, all the calculations in the proof of Theorem~1 in \cite{Devetak-Shor-2005}
follow in this case as well. 
\end{proof}

\section{Discussion}\label{sec: Discussion}
In this paper, we introduce a new quantity called the generalized information that mathematically generalizes the classical mutual information and the coherent information. We also define a task where 
the goal is to transmit maximal number of copies of system $A'$ of a general mixed state $\rho^{A'R}$ over $n$ uses of a quantum channel $\cN$. We define the generalized capacity as $mS(CQ)_{\omega}/n$, i.e. the number of copies of $\rho^{A'R}$ being transmitted  per channel use times the entropy of the classical and quantum systems in the Koashi-Imoto decomposition of $A'$.
We characterize the generalized capacity of a quantum channel, and show that the generalized capacity is equal  to the regularization of the 
generalized information. Unfortunately, the generalized capacity inherits its  multi-letter nature from the Holevo capacity and the quantum capacity. Namely, since 
the generalized capacity reduces to transmitting classical information or entanglement over a quantum channel in special cases, as explained in Remark~\ref{remark: special cases}, its rate cannot be simplified unless 
for specific channels for which the Holevo capacity or the quantum capacity rates are
single-letter.

The generalized capacity has an interesting geometric interpretation, and is directly related to the simultaneous transmission of classical and quantum information studied by Devetak and Shor in \cite{Devetak-Shor-2005}, where the optimal trade-off between the classical and quantum rates is found. The trade-off curve is illustrated in Fig.~\ref{fig: Thm_diagram}, and we call it the Devetak-Shor curve.
%
We can find the generalized capacity of a quantum channel $\cN$ for transmitting a state $\rho^{A'R}$, with KI decomposition $\omega_{\rho}^{CNQR}$ using the Devetak-Shor curve, as shown in Fig~\ref{fig: Thm_diagram}. 
We draw the line $R_C=\frac{S(C)_{\omega}}{S(Q|C)_{\omega}} R_Q$ and find the intersection of this line with the Devetak-Shor curve. Let $(R_Q^*,R_C^*)$ denote the intersected point. 
Then, the generalized capacity is $C_G(\rho^{A'R},\cN)=R_Q^*+R_C^*$. This implies that the maximal number of copies of the state $\rho^{A'R}$, which can be transmitted per channel use in the asymptotic limit, is  $\frac{R_Q^*+R_C^*}{S(CQ)_{\omega}}$. For a fully classical $\rho^{A'R}$ with $S(Q|C)_{\omega}=0$, the slop of the line is $\infty$, so the generalized capacity is equal to the classical capacity as $C_G(\rho^{A'R},\cN)=C_C$.
For a pure source $\ket{\rho}^{A'R}$  with $S(C)_{\omega}=0$,  the slop of the line is $0$, hence, the generalized capacity is equal to the quantum capacity as $C_G(\rho^{A'R},\cN)=C_Q$. The fact that the Devetak-Shor curve is below the line 
$R_C=C_C-R_Q$ implies that the generalized capacity, which equals to the sum $R_Q^*+R_C^*$, is maximized for a state $\rho^{A'R}$ when $A'$ and $R$ are both classical systems.

%
At the end of Sec.~\ref{sec: transmission task}, we discuss that to make the problem technically easier to deal with, we can consider the transmission of $\omega_{\rho}^{CQR}$, i.e. the KI-decomposition instead of the state  $\rho^{A'R}$ since they are 
equivalent sources in the sense that they are related via CPTP maps in both directions. 
Then, we construct a code to transmit the classical system $C^m$ and the quantum system $Q^m$
of $(\omega_{\rho}^{CQR})^{\ox m}$. 
For this, we apply Lemma~\ref{lemma: DS extension}, which is an extension of the protocol designed by Devetak and Shor for simultaneous transmission of classical and quantum information as a sub-protocol \cite{Devetak-Shor-2005}.

\bigskip

\textbf{Acknowledgments.}
I thank  Andreas Winter and Patrick Hayden for insightful discussions and comments on the initial draft.  
The author was supported by the Marie Sk{\l}odowska-Curie Actions (MSCA)
European Postdoctoral Fellowships (Project 101068785-QUARC).

\appendix
\section{The fidelity equivalence} 

\begin{lemma}\label{lemma: almost pure reduced state}
    Let $\xi^{AB}$ be a state with the reduced state on system $A$ being almost pure as the following 
    \begin{align*}
        F(\xi^A,\proj{\psi}^A)\geq 1-\epsilon. 
    \end{align*}
     Then, the state $\xi^{AB}$ is an almost product state in the following sense
    \begin{align*}
        &F(\xi^{AB},\proj{\psi}^A \ox \xi^B)\geq 1-4\epsilon.
    \end{align*}
\end{lemma}
\begin{proof}
From the fidelity criterion on system $A$ we obtain
\begin{align}\label{eq: operator norm xi_A} 
    1-\epsilon & \leq F(\xi^A,\proj{\psi}^A) \nonumber\\
    &=\sqrt{\bra{\psi} \xi^A \ket{\psi}} \nonumber \\
    &\leq \sqrt{ \norm{\xi^A}_{\infty}},
\end{align}
where the last line follows from the definition of the operator norm.
Now, consider a purification of $\xi^{AB}$ with the following Schmidt decomposition 
\begin{align*}
\ket{\xi}^{ABR}
     = \sum_{i} \sqrt{\lambda_i}\ket{v_i}^{A} \ket{w_i}^{BR}. 
\end{align*}
Considering the above decomposition, we obtain
\begin{align}
  \label{eq:almost-pure}
   F\left( {\xi}^{AB},
                             {\xi}^{A} \ox {\xi}^{B} \right)
   &\geq F\left( \ketbra{\xi}^{ABR},
                             {\xi}^{A} \ox {\xi}^{BR} \right) \nonumber\\
    &= \sqrt{\bra{\xi}{\xi}^{A} \ox {\xi}^{BR} 
                               \ket{\xi}}                                                                 \nonumber\\
    &= \sum_i \lambda_i^{\frac32} \nonumber\\
    & \geq \norm{\xi^{A}}_{\infty}^{\frac32} \nonumber\\
    &\geq (1-\epsilon)^3 
     \geq 1 - 3\epsilon,   
\end{align}
where the penultimate inequality follows since the sum of the eigenvalues of an operator is bigger than its largest eigenvalue. 
The last line follows from Eq.~(\ref{eq: operator norm xi_A}).
We also note that 
\begin{align}
    F(\xi^{A}\ox \xi^B,\proj{\psi}^A \ox \xi^B)=F(\xi^{A},\proj{\psi}^A) \nonumber
    \geq 1-\epsilon.
\end{align}
We use the above equation and Eq.~(\ref{eq:almost-pure}) to compute the fidelity
$F(\xi^{AB},\proj{\psi}^A \ox \xi^B)$ by
applying the Fuchs-van de Graaf inequality in Eq.~(\ref{eq: Fuchs-van de Graaf})  and the triangle inequality to the trace norm.
\end{proof}

\section{Proof of Lemma~\ref{lemma:Y_epsilon properties}}\label{sec: Proof of Lemma Y W epsilon}
\begin{enumerate}
\item The definition of the functions $Y_{\epsilon}(\omega)$ and $W_{\epsilon}(\omega)$
  directly implies that they are non-decreasing functions of $\epsilon$.

\item We first prove the concavity of $Y_{\epsilon}(\omega)$. 
  Let $U_1:C Q \hookrightarrow \hat{C}  \hat{Q} E$ and 
  $U_2:C  Q \hookrightarrow \hat{C} \hat{Q} E$ be the isometries attaining the 
  maximum for $\epsilon_1$ and $\epsilon_2$, respectively, which act as 
  follows on the purification $\ket{\omega}^{C Q R R' C' C''}$ of the state $\omega^{C N Q R R' C'}$ with $C''$ as the purifying system
  \begin{align}
    \ket{\tau_1}^{\hat{C}  \hat{Q} E R R' C' C''}
        &=(U_1 \otimes \1_{R R' C' C''}) \ket{\omega}^{C  Q R R' C' C''} \nonumber\\
    \ket{\tau_2}^{\hat{C}  \hat{Q} E R R'C' C''}
        &=(U_2 \otimes \1_{R R'C' C''}) \ket{\omega}^{C Q R R' C' C''}, \nonumber  
  \end{align}
  where $\Tr_{C''}[\proj{\omega}^{C  Q R R'C' C''}]=\omega^{C Q R R' C'}$. 
  For $0\leq \lambda \leq 1$, define the isometry 
  $U_0:C Q \hookrightarrow \hat{C} \hat{Q} E F F'$ which acts as 
  \begin{equation}
    \label{eq: isometry U in convexity}
    U_0 := \sqrt{\lambda} U_1 \otimes \ket{11}^{FF'} + \sqrt{1-\lambda} U_2 \otimes \ket{22}^{FF'},
  \end{equation}
  where systems $F$ and $F'$ are qubits, and
  which leads to the state
  \begin{align}
    \ket{\tau_0}^{\hat{C} \hat{Q} EFF' RR' C' C''}&:=(U_0 \otimes \1_{R R' C' C''}) \ket{\omega}^{C  Q RR' C' C''} \nonumber\\
      &= \sqrt{\lambda}\ket{\tau_1}^{\hat{C}  \hat{Q} E RR' C' C''} \ket{11}^{FF'}
        + \sqrt{1-\lambda}\ket{\tau_2}^{\hat{C} \hat{Q} E RR' C' C''} \ket{22}^{FF'}.   
  \end{align}
  Then, $U_0$ defines its state $\tau_0$, for which the reduced state on the systems 
  $\hat{C} \hat{Q}  R $ is 
  \begin{align} \label{eq: tau in convexity proof}
    \tau_0^{\hat{C} \hat{Q} R} 
      =\lambda \tau_1^{\hat{C} \hat{Q} R }+ (1-\lambda) \tau_2^{\hat{C} \hat{Q} R}. 
  \end{align}  
  Therefore, the fidelity for the state $\tau$ is bounded as follows:
  \begin{align}\label{eq:fidelity in convexity}
    F(\omega^{C  Q R} ,\tau_0^{\hat{C}  \hat{Q} R} )
      &= F(\omega^{C  Q R} ,\lambda \tau_1^{\hat{C} \hat{Q} R}
        + (1-\lambda) \tau_2^{\hat{C}   \hat{Q} R}) \nonumber \\
      &= F(\lambda \omega^{C Q R}+(1-\lambda)\omega^{C  Q R},
           \lambda \tau_1^{\hat{C}  \hat{Q} R}
            + (1-\lambda) \tau_2^{\hat{C}  \hat{Q} R}) \nonumber\\
      &\geq \lambda F( \omega^{C Q R},\tau_1^{\hat{C}  \hat{Q} R})
            +(1-\lambda)F( \omega^{C  Q R},\tau_2^{\hat{C}  \hat{Q} R}) \nonumber\\
     &\geq 1-\left( \lambda\epsilon_1 +(1-\lambda)\epsilon_2 \right).
  \end{align}
  The first inequality is due to simultaneous concavity of the fidelity in both
  arguments;
  the last line follows by the definition of the isometries $U_1$ and $U_2$.
  Thus, the isometry $U_0$ yields a fidelity of at least 
  $1-\left( \lambda\epsilon_1 +(1-\lambda)\epsilon_2 \right) =: 1-\epsilon$.
  According to Definition \ref{def:Y_epsilon}, we obtain
  \begin{align}
    Y_\epsilon(\omega) &\geq S( \hat{Q} RR'|\hat{C})_{\tau_0} \nonumber\\
                       &\geq S(\hat{Q} RR'|\hat{C} F)_{\tau_0} \nonumber\\
                       &= \lambda S(  \hat{Q} RR'|\hat{C} )_{\tau_1}+(1-\lambda)S(\hat{Q} RR'|\hat{C} )_{\tau_2} \nonumber\\
                       &= \lambda Y_{\epsilon_1}(\omega)+(1-\lambda)Y_{\epsilon_2}(\omega),
  \end{align}
  where the second line is due to strong sub-additivity of the entropy. 
  The third and the last lines follow from the definition of $\tau_0$ and $\tau_1,\tau_2$,
  respectively. \\
  To prove the concavity of $W_{\epsilon}(\omega)$, we take exactly the same above steps, and then by 
  Definition \ref{def:Y_epsilon}, we obtain
  \begin{align}
    W_\epsilon(\omega) &\geq S( \hat{C} |C')_{\tau_0} \nonumber\\
                       &\geq S(\hat{C} |C' F)_{\tau_0} \nonumber\\
                       &= \lambda S( \hat{C} |C' )_{\tau_1}+(1-\lambda)S(\hat{C} |C' )_{\tau_2} \nonumber\\
                       &= \lambda W_{\epsilon_1}(\omega)+(1-\lambda)W_{\epsilon_2}(\omega),
  \end{align}
  where the second line is due to strong sub-additivity of the entropy. 
  The third and the last lines follow from the definition of $\tau_0$ and $\tau_1,\tau_2$,
  respectively.
\item The functions $Y_{\epsilon}(\omega)$ and $W_{\epsilon}(\omega)$ are concave for $\epsilon \geq 0 $, so it is 
  continuous 
  for $\epsilon > 0$. 
  The concavity implies furthermore that they are lower semi-continuous at 
  $\epsilon=0$. On the other hand, since the fidelity and the conditional entropy are  continuous functions of CPTP maps, and the domain both optimizations 
  is a compact set, we conclude that $Y_\epsilon(\omega)$ and $W_\epsilon(\omega)$ are also upper 
  semi-continuous at $\epsilon=0$, so they are continuous at $\epsilon=0$ 
  \cite[Thms.~10.1 and 10.2]{Rockafeller}. 

\item We  prove 
  $Y_{\epsilon}(\omega_1 \otimes \omega_2) \leq  Y_{\epsilon}(\omega_1) +Y_{\epsilon}(\omega_2)$.
  In the definition of $Y_{\epsilon}(\omega_1 \otimes \omega_2)$, let the isometry 
  $U_0:C_1 Q_1 C_2 Q_2 \hookrightarrow \hat{C}_1 \hat{Q}_1 \hat{C}_2 \hat{Q}_2 E$
  be the one attaining the maximum, which acts on the following purification of the extended source states with purifying 
  systems $C''_1$ and $C''_2$: 
  \begin{equation}
    \label{eq:U0-action}
    \ket{\tau}^{\hat{C}_1 \hat{Q}_1 \hat{C}_2 \hat{Q}_2 E R_1 R'_1 C'_1 C''_1  R_2 R'_2 C'_2  C''_2}
               =(U_0 \otimes \1_{R_1 R'_1 C'_1 C''_1  R_2 R'_2 C'_2  C''_2})\ket{\omega_1}^{C_1 Q_1 R_1 R'_1 C'_1 C''_1}
                                                   \otimes    \ket{\omega_2}^{C_2 Q_2 R_2 R'_2 C'_2  C''_2}.
  \end{equation}
  By the definition, the fidelity is bounded by
  \begin{align*}
    F(\omega_1^{C_1 Q_1 R_1} \otimes \omega_2^{C_2 Q_2 R_2},
      \tau^{\hat{C}_1  \hat{Q}_1 \hat{C}_2   \hat{Q}_2 R_1 R_2}) \geq 1- \epsilon.   
  \end{align*}
  We can define an isometry 
  $U_1:C_1 Q_1 \hookrightarrow \hat{C}_1 \hat{Q}_1 E_1$ 
  acting only on systems $C_1 Q_1$, by letting
  $U_1 = (U_0 \otimes \1_{R_2 R_2' C_2' C''_2 })(\1_{C_1 Q_1} \otimes \ket{\omega_2}^{C_2 Q_2 R_2 R'_2 C_2' C''_2 })$
  and with the environment $E_1 := \hat{C}_2 \hat{Q}_2 E R_2 R'_2 C'_2 C''_2 $.
  This implies that 
  $\ket{\tau_1}^{\hat{C}_1  \hat{Q}_1 R_1 R_1' C_1' C''_1  E_1} 
   := (U_1 \otimes \1_{R_1R_1' C_1' C''_1 })\ket{\omega_1}^{C_1 Q_1 R_1 R_1' C_1' C''_1}$ 
  has the same reduced state on $\hat{C}_1 \hat{Q}_1 R_1$ as $\tau$ from
  Eq. (\ref{eq:U0-action}).
  This isometry preserves the fidelity for $\omega_1$, which follows from monotonicity 
  of the fidelity under partial trace:
  \begin{align*}
     F(\omega_1^{C_1 Q_1 R_1},\tau_1^{\hat{C}_1  \hat{Q}_1 R_1}) 
       &= F(\omega_1^{C_1  Q_1 R_1},\tau^{\hat{C}_1 \hat{Q}_1 R_1}) \\
       &\geq F(\omega_1^{C_1  Q_1 R_1} \otimes \omega_2^{C_2  Q_2 R_2},
               \tau^{\hat{C}_1  \hat{Q}_1 \hat{C}_2  \hat{Q}_2 R_1 R_2}) \\
       &\geq 1- \epsilon.   
  \end{align*}
  By the same argument, there is an isometry 
  $U_2:C_2  Q_2\hookrightarrow \hat{C}_1 \hat{Q}_1 \hat{C}_2  \hat{Q}_2 E R_1 R'_1 C'_1 C''_1$ and state $\ket{\tau_2}^{\hat{C}_2  \hat{Q}_2 R_2 R_2' C_2' C''_2  E_2} 
   := (U_2 \otimes \1_{R_2R_2' C_2' C''_2 })\ket{\omega_2}^{C_2 Q_2 R_2 R_2' C_2' C''_2}$
  with output system $\hat{C}_2  \hat{Q}_2$ and 
  environment $E_2:=\hat{C}_1 \hat{Q}_1 E R_1 R'_1 C'_1 C''_1$, such that
  \begin{align*}
    F(\omega_2^{C_2  Q_2 R_2},\tau_2^{\hat{C}_2 \hat{Q}_2 R_2}) 
      &=    F(\omega_2^{C_2  Q_2 R_2},\tau^{\hat{C}_2  \hat{Q}_2 R_2}) \\
      &\geq F(\omega_1^{C_1  Q_1 R_1} \otimes \omega_2^{C_2  Q_2 R_2},
              \tau^{\hat{C}_1  \hat{Q}_1 \hat{C}_2  \hat{Q}_2 R_1 R_2}) \\
      &\geq 1- \epsilon.   
  \end{align*}
  Therefore, we obtain:
  \begin{align}
     Y_{\epsilon}&(\omega_1) +Y_{\epsilon}(\omega_2)-Y_{\epsilon}(\omega_1 \otimes \omega_2) \nonumber\\
     &\geq
     S(\hat{Q}_1 R_1R'_1|\hat{C}_1)_{\tau_1}+S(\hat{Q}_2 R_2R'_2|\hat{C}_2)_{\tau_2}-S(\hat{Q}_1 \hat{Q}_2 R_1R'_1 R_2R'_2|\hat{C}_1\hat{C}_2)_{\tau} \nonumber \\
     &=
     S(\hat{Q}_1 R_1R'_1|\hat{C}_1)_{\tau}+S(\hat{Q}_2 R_2R'_2|\hat{C}_2)_{\tau}-S(\hat{Q}_1 \hat{Q}_2 R_1R'_1 R_2R'_2|\hat{C}_1\hat{C}_2)_{\tau} \nonumber \\
     &=S(\hat{Q}_1 R_1R'_1|\hat{C}_1)_{\tau}+S(\hat{Q}_2 R_2R'_2|\hat{C}_2)_{\tau}-S(\hat{Q}_1 R_1R'_1 |\hat{C}_1\hat{C}_2)_{\tau}
     -S( \hat{Q}_2  R_2R'_2|\hat{C}_1\hat{C}_2\hat{Q}_1 R_1R'_1)_{\tau} \nonumber\\
     &\geq 0 \nonumber,
  \end{align}
  where the first line is due to Definition~\ref{def:Y_epsilon}.
  The second line follows from the definition of $\tau_1$ and $\tau_2$ which satisfy
    $\tau_1^{\hat{C}_1  \hat{Q}_1 R_1 R_1' C_1' C''_1 }=\tau^{\hat{C}_1  \hat{Q}_1 R_1 R_1' C_1' C''_1 }$ and $\tau_2^{\hat{C}_2  \hat{Q}_2 R_2 R_2' C_2' C''_2 }=\tau^{\hat{C}_2  \hat{Q}_2 R_2 R_2' C_2' C''_2 }$.
  The third line is due to the chain rule.
  The last line follows from strong sub-additivity of the entropy.\\ 
  To prove $W_{\epsilon}(\omega_1 \otimes \omega_2) \leq  W_{\epsilon}(\omega_1) +W_{\epsilon}(\omega_2)$, we take exactly the same steps as above, then we obtain
  \begin{align}
     W_{\epsilon}&(\omega_1) +W_{\epsilon}(\omega_2)-W_{\epsilon}(\omega_1 \otimes \omega_2) \nonumber\\
     &\geq
     S(\hat{C}_1|C'_1)_{\tau_1}+S(\hat{C}_2|C'_2)_{\tau_2}-S(\hat{C}_1 \hat{C}_2|C'_1 C'_2)_{\tau} \nonumber \\
     &=
     S(\hat{C}_1 |C'_1)_{\tau}+S(\hat{C}_2|C'_2)_{\tau}-S(\hat{C}_1 \hat{C}_2|C'_1 C'_2)_{\tau} \nonumber \\
     &=S(\hat{C}_1|C'_1)_{\tau_1}\!\!+\!S(\hat{C}_2|C'_2)_{\tau_2}\!\!\!-\!S(\hat{C}_1 |C'_1 C'_2)_{\tau}
     \!\!-\!S( \hat{C}_2|C'_1 C'_2 \hat{C}_1)_{\tau} \nonumber\\
     &\geq 0 \nonumber,
  \end{align}
  where the first line is due to Definition~\ref{def:Y_epsilon}.
  The second line follows from the definition of $\tau_1$ and $\tau_2$ which satisfy
    $\tau_1^{\hat{C}_1  C_1' }=\tau^{\hat{C}_1  C_1'}$ and $\tau_2^{\hat{C}_2  C'_2 }=\tau^{\hat{C}_2  C'_2 }$.
  The third line is due to the chain rule.
  The last line follows from strong sub-additivity of the entropy. 
\item According to Theorem~\ref{thm: KI decomposition} \cite{KI2002,Hayden2004}, 
  any isometry $U:C N Q \rightarrow \hat{C}  \hat{Q} E$ acting on the state 
  $\omega^{C Q R C'}$ which preserves the reduced state on systems $C Q R $ with fidelity equal to 1, acts as an identity on systems $C$ and $Q$
  \begin{align*}
    (\1_{CQ} \otimes \1_{RC'}) \omega^{C  Q R C'}(\1_{CQ}^{\dagger} \otimes \1_{RC'})
      =\sum_{c} p_c \proj{c}^{C}\otimes  \omega_{c}^{Q R} \otimes \proj{c}^{C'}.
  \end{align*}
  Therefore,  in Definition~\ref{def:Y_epsilon} for $\epsilon=0$, the final state is
  \begin{align*}
    \tau^{\hat{C}  \hat{Q}   R R' C'}
      = \sum_{c} p_c \proj{c}^{C}\ox \proj{\omega_{c}}^{Q R R'} \otimes \proj{c}^{C'}.
  \end{align*} 
  Thus, we can directly evaluate
\begin{align*}
      Y_0(\omega)&=S(\hat{Q} RR'|\hat{C})_\tau=S(Q RR'|C)_\omega=0,\\
      W_0(\omega)&=S(\hat{C} |C')_\tau=S(C|C')_\omega=0,
\end{align*}
concluding the proof.
\hfill\qedsymbol
\end{enumerate}

\section{The set $\cS_G^{(l)}$ of points $(I(R\> \rangle B^lX)_{\sigma}/l,I(B^l:X)_{\sigma}/l)$}\label{sec: single-letter curve}
For a quantum channel $\cN:A \to B$ and a natural number $l$,  we define the set $\cS_G^{(l)}$ of points $(R_Q,R_C)=(I(R\> \rangle B^lX)_{\sigma}/l,I(B^l:X)_{\sigma}/l)$ arising from applying the channel $\cN^{\ox l}$ to different states 
$\rho^{A^lRX}$ such that $I(R\> \rangle B^lX)_{\sigma} \geq 0$.
We note that the coherent information can be negative, however,  we are only interested 
in non-negative values, so we define the set only for states with $R_Q \geq 0$.
In the limit of many copies of the channel, the set $\cS_G^{(l)}$ is equal to the capacity region in Fig.~\ref{fig: DS curve}.
Here, we would like to understand the set $\cS_G^{(l)}$.

\begin{figure}[t] 
  \includegraphics[width=0.5\textwidth]{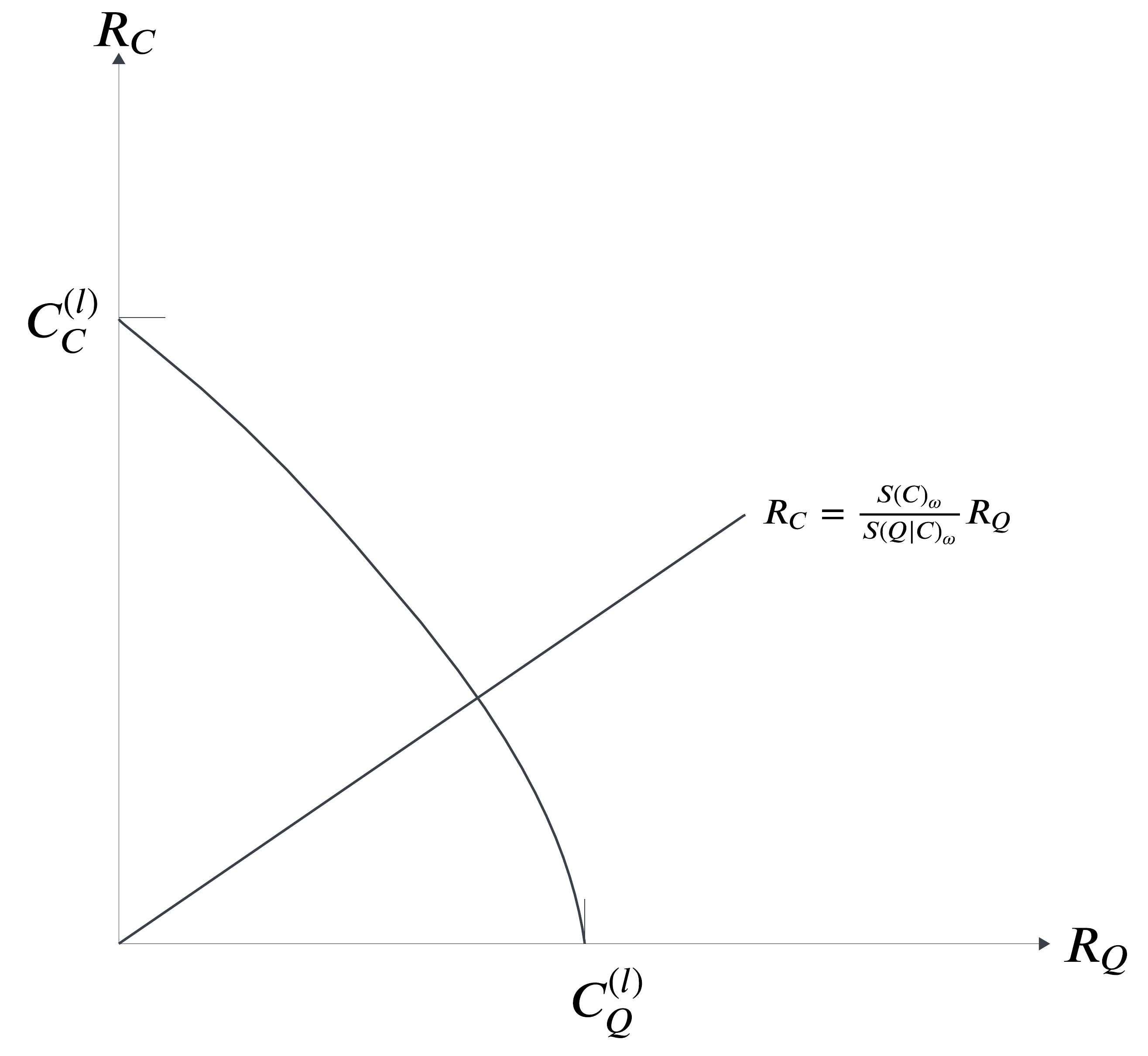} 
  \caption{The set $\cS_G^{(l)}$ is the area under the curve with $R_Q\geq 0$ and $R_C \geq 0$}
  \label{fig: S_G}
\end{figure}

We note that the maximum in $\max_{\rho^{A^lRX}}I(R\> \rangle B^lX)_{\sigma}/l$ is attained for a pure state $\ket{\varphi_0}^{A^lR}$. This follows since this coherent information is a convex combination as $I(R\> \rangle B^lX)_{\sigma}/l=\sum_x p(x) I(R\> \rangle B^l)_{\sigma_x}/l$, and one of the terms in any convex combination is greater than any other term. 
Therefore, there is a pure state $\ket{\varphi_0}^{A^lR}$ for which the maximum is attained.
This state with $X=\emptyset$ corresponds to the point  $(C_Q^{(l)},0)$ on $\cS_G^{(l)}$.

The maximum in $\max_{\rho^{A^lRX}}I(B^l:X)_{\sigma}/l$ is attained for a state
with $R=\emptyset$, i.e. $\varphi_1^{A^lX}=\sum_x p(x)\proj{\varphi_{1x}}^{A^l}\ox \proj{x}^X$. To see this, first trace out system $R$ since $I(B^l:X)_{\sigma}$ does not depend on system $R$, and
assume that $\rho_1^{A^lX}$ is the state attaining the maximum, and the state $\sigma_1$ arises from applying the channel as 
\begin{align}
    \rho_1^{A^lX}&=\sum_x p(x){\rho_{1x}}^{A^l}\ox \proj{x}^X, \nonumber\\
    \sigma_1^{B^lX}&=(\cN^{\ox l}\ox\id_X) \rho_1^{A^lX} 
    =\sum_x p(x){\sigma_{1x}}^{B^l}\ox \proj{x}^X. \nonumber
\end{align}
For a given $x$, consider the spectral decomposition of ${\rho_{1x}}^{A^l}=\sum_{y} p(y|x)\proj{\rho_{1xy}}^{A^l}$ and define  extensions of the states $\rho_1^{A^lX}$ and $\sigma_1^{B^lX}$ as
\begin{align}
    \rho_1^{A^lXY}
    &=\sum_{x,y} p(x)p(y|x)\proj{\rho_{1xy}}^{A^l}\ox \proj{x}^X \ox \proj{y}^Y, \nonumber\\
    \sigma_1^{A^lXY}
    &=(\cN^{\ox l}\ox  \id_{XY})\rho_1^{A^lXY}=\sum_{x,y} p(x)p(y|x){\sigma_{1xy}}^{B^l}\ox \proj{x}^X \ox \proj{y}^Y. \nonumber
\end{align}
For the above state $I(B^l:X)_{\sigma_1} \leq I(B^l:XY)_{\sigma_1}$ holds. Therefore, 
w.l.o.g. we can assume that the state attaining the maximum is of the form $\varphi_1^{A^lX}=\sum_x p(x)\proj{\varphi_{1x}}^{A^l}\ox \proj{x}^X$.
This state with $R=\emptyset$ corresponds to the point  $(0,C_C^{(l)})$ on $\cS_G^{(l)}$.

We claim that the set $\cS_G^{(l)}$ is convex as shown in Fig.~\ref{fig: S_G}.
To see this, for two given states $\rho_1^{A^lRX}$, $\rho_2^{A^lRX}$ and $0 \leq \lambda \leq 1$ define
\begin{align}
    \rho_3^{A^lRX}=\lambda \rho_1^{A^lRX} \ox \proj{1}^Y+(1-\lambda)\rho_2^{A^lRX}\ox \proj{1}^Y. \nonumber
\end{align}
By applying the channel on the above state, we obtain 
\begin{align}
    I(R\> \rangle B^lXY)_{\sigma_3}/l&=\lambda I(R\> \rangle B^lX)_{\sigma_1}/l+
    (1-\lambda) I(R\> \rangle B^lX)_{\sigma_2}/l \nonumber\\
    I(B^l:XY)_{\sigma_3}/l& \geq \lambda I(B^l:X)_{\sigma_1}/l+
    (1-\lambda) I(B^l:X)_{\sigma_2}/l. \nonumber
\end{align}

The set $\cS_G^{(l)}$ is illustrated in Fig.~\ref{fig: S_G}. For states
$\rho^{A^lRX}$ giving rise to the points on the line $R_C=\frac{S(C)_{\omega}}{S(Q|C)_{\omega}} R_Q$
the ratio 
$\frac{I(B^l:X)_{\sigma}}{I(R\> \rangle B^lX)_{\sigma}}=\frac{S(C)_{\omega}}{S(Q|C)_{\omega}}$
holds.

\section{Miscellaneous definitions and facts}\label{sec: Miscellaneous definitions and facts}
\begin{definition}\label{def:typicality}
Let $X$ be a random variable with probability distribution $\{p(x)\}_{x \in \cX}$. For $\delta >0$, define the typical set as
\begin{align} 
    \mathcal{T}_{\delta,X}^n:=\left\{x^n=x_1 x_2 \ldots x_n : \forall x \abs{N(x|x^n)-np(x)} \leq n\delta\right\}, \nonumber
\end{align}
where $N(x|x^n)$ counts the number of occurrences of $x$ in the
sequence $x^n$ of length $n$.
Moreover, let $Y$ be a random variable correlated with $X$, with the joint probability distribution $\{p(x,y)\}_{(x,y)\in(\cX,\cY)}$, and 
let $p(y|x)$ denote the conditional probability distribution satisfying $p(x,y)=p(x)p(y|x)$ for all $x$ and $y$.
For a typical $x^n \in \mathcal{T}_{\delta,\rho }^n$, 
define the
set of conditionally typical sequences (with $x^n$) by
\begin{align} 
    \mathcal{T}_{\delta,Y|X}^n(x^n):=\left\{y^n=y_1 y_2 \ldots y_n : \forall x,y \quad\abs{N(x,y|x^n,y^n)-p(y|x)N(x|x^n)} \leq n\delta\right\}. \nonumber
\end{align}

\noindent For state $\rho^A$ with spectral decomposition $\rho^A=\sum_x p(x)\proj{e_x}^A$ and constant $\delta$, define the typical projector of $\rho^{\ox n}$   as
\begin{align*}
    \Pi^n_{\delta ,\rho }:=\sum_{x^n \in \mathcal{T}_{\delta,\rho}^n} \proj{e_{x^n}}
    =\sum_{x^n \in \mathcal{T}_{\delta,\rho}^n} \proj{e_{x_1}} \ox \proj{e_{x_2}} \ox \cdots
    \proj{e_{x_n}}.
\end{align*}
Also, Let  $\sigma^{XB}=\sum_x p(x)\proj{x}^X \ox \sigma_x^B$ be a cq state with the spectral decomposition of $\sigma_x^B$ as $\sigma_x^B=\sum_y p(y|x) \proj{w_{x,y}}^B$. 
For a typical $x^n \in \mathcal{T}_{\delta,\rho }^n$, define the conditional typical projector of $\sigma_{x^n}=\sigma_{x_1} \ox \cdots \sigma_{x_n}$  as
\begin{align*}
    \Pi^n_{\delta ,\sigma_{x^n} }:=\sum_{y^n \in\mathcal{T}_{\delta,Y|X}^n(x^n)} \proj{w_{x^n,y^n}}
    =\sum_{y^n \in\mathcal{T}_{\delta,Y|X}^n(x^n)} \proj{w_{x_1,y_1}} \ox \proj{w_{x_2,y_2}} \ox \cdots
    \proj{w_{x_n,y_n}}.
\end{align*}

\end{definition}

\begin{lemma}[{Cf.~\cite{csiszar_korner_2011}}]
\label{lemma:typicality properties}
For any $\delta >0$ the following properties hold
\begin{align*}
    \Pr(x^n \in \mathcal{T}_{\delta,X}^n) = \Tr(\rho^{\ox n}\Pi^n_{\delta ,\rho })&\geq 1-\epsilon, \\
    \Tr(\sigma_{x^n}\Pi^n_{\delta ,\sigma_{x^n} })&\geq 1-\epsilon,\\
    2^{n[S(A)_{\rho}+c\delta]} &\geq \Tr(\Pi^n_{\delta ,\rho }),\\
     2^{n[S(B|X)_{\sigma}+c\delta]} &\geq \Tr(\Pi^n_{\delta ,\sigma_{x^n}}),
\end{align*}
for some constant $c$ and  $\epsilon = 2^{-n c' \delta^n}$ for some constant $c'$. 
\end{lemma}

\begin{lemma}[{Fannes~\cite{Fannes1973}; Audenaert~\cite{Audenaert2007}}]
\label{Fannes-Audenaert inequality}
Let $\rho$ and $\sigma$ be two states on Hilbert space 
$A$ with trace distance 
$\frac12\|\rho-\sigma\|_1 \leq \epsilon$, then
\begin{align*}
  |S(\rho)-S(\sigma)| \leq \epsilon\log |A| + h(\epsilon),
\end{align*}
where $h(\epsilon)=-\epsilon \log \epsilon -(1-\epsilon)\log (1-\epsilon)$ is the binary entropy.
\end{lemma}

\bibliography{capacity_references}

\begin{thebibliography}{19}%
\makeatletter
\providecommand \@ifxundefined [1]{%
 \@ifx{#1\undefined}
}%
\providecommand \@ifnum [1]{%
 \ifnum #1\expandafter \@firstoftwo
 \else \expandafter \@secondoftwo
 \fi
}%
\providecommand \@ifx [1]{%
 \ifx #1\expandafter \@firstoftwo
 \else \expandafter \@secondoftwo
 \fi
}%
\providecommand \natexlab [1]{#1}%
\providecommand \enquote  [1]{``#1''}%
\providecommand \bibnamefont  [1]{#1}%
\providecommand \bibfnamefont [1]{#1}%
\providecommand \citenamefont [1]{#1}%
\providecommand \href@noop [0]{\@secondoftwo}%
\providecommand \href [0]{\begingroup \@sanitize@url \@href}%
\providecommand \@href[1]{\@@startlink{#1}\@@href}%
\providecommand \@@href[1]{\endgroup#1\@@endlink}%
\providecommand \@sanitize@url [0]{\catcode `\\12\catcode `\$12\catcode
  `\&12\catcode `\#12\catcode `\^12\catcode `\_12\catcode `\%12\relax}%
\providecommand \@@startlink[1]{}%
\providecommand \@@endlink[0]{}%
\providecommand \url  [0]{\begingroup\@sanitize@url \@url }%
\providecommand \@url [1]{\endgroup\@href {#1}{\urlprefix }}%
\providecommand \urlprefix  [0]{URL }%
\providecommand \Eprint [0]{\href }%
\providecommand \doibase [0]{http://dx.doi.org/}%
\providecommand \selectlanguage [0]{\@gobble}%
\providecommand \bibinfo  [0]{\@secondoftwo}%
\providecommand \bibfield  [0]{\@secondoftwo}%
\providecommand \translation [1]{[#1]}%
\providecommand \BibitemOpen [0]{}%
\providecommand \bibitemStop [0]{}%
\providecommand \bibitemNoStop [0]{.\EOS\space}%
\providecommand \EOS [0]{\spacefactor3000\relax}%
\providecommand \BibitemShut  [1]{\csname bibitem#1\endcsname}%
\let\auto@bib@innerbib\@empty
\bibitem [{\citenamefont {Shor}(1995)}]{decoherence-Shor-1995}%
  \BibitemOpen
  \bibfield  {author} {\bibinfo {author} {\bibfnamefont {Peter~W.}\
  \bibnamefont {Shor}},\ }\bibfield  {title} {\enquote {\bibinfo {title}
  {Scheme for reducing decoherence in quantum computer memory},}\ }\href@noop
  {} {\bibfield  {journal} {\bibinfo  {journal} {Phys. Rev. A}\ }\textbf
  {\bibinfo {volume} {52}},\ \bibinfo {pages} {R2493--R2496} (\bibinfo {year}
  {1995})}\BibitemShut {NoStop}%
\bibitem [{\citenamefont {Schumacher}(1996)}]{Schumi1996}%
  \BibitemOpen
  \bibfield  {author} {\bibinfo {author} {\bibfnamefont {B.}~\bibnamefont
  {Schumacher}},\ }\bibfield  {title} {\enquote {\bibinfo {title} {Sending
  entanglement through noisy quantum channels},}\ }\href@noop {} {\bibfield
  {journal} {\bibinfo  {journal} {Phys. Rev. A}\ }\textbf {\bibinfo {volume}
  {54}},\ \bibinfo {pages} {2614--2628} (\bibinfo {year} {1996})}\BibitemShut
  {NoStop}%
\bibitem [{\citenamefont {Schumacher}\ and\ \citenamefont
  {Nielsen}(1996)}]{err-correction-Schumacher-1996}%
  \BibitemOpen
  \bibfield  {author} {\bibinfo {author} {\bibfnamefont {Benjamin}\
  \bibnamefont {Schumacher}}\ and\ \bibinfo {author} {\bibfnamefont {M.~A.}\
  \bibnamefont {Nielsen}},\ }\bibfield  {title} {\enquote {\bibinfo {title}
  {Quantum data processing and error correction},}\ }\href@noop {} {\bibfield
  {journal} {\bibinfo  {journal} {Phys. Rev. A}\ }\textbf {\bibinfo {volume}
  {54}},\ \bibinfo {pages} {2629--2635} (\bibinfo {year} {1996})}\BibitemShut
  {NoStop}%
\bibitem [{\citenamefont {Barnum}\ \emph {et~al.}(1998)\citenamefont {Barnum},
  \citenamefont {Nielsen},\ and\ \citenamefont {Schumacher}}]{Barnum1998}%
  \BibitemOpen
  \bibfield  {author} {\bibinfo {author} {\bibfnamefont {H.}~\bibnamefont
  {Barnum}}, \bibinfo {author} {\bibfnamefont {M.~A.}\ \bibnamefont {Nielsen}},
  \ and\ \bibinfo {author} {\bibfnamefont {B.~W.}\ \bibnamefont {Schumacher}},\
  }\bibfield  {title} {\enquote {\bibinfo {title} {Information transmission
  through a noisy quantum channel},}\ }\href@noop {} {\bibfield  {journal}
  {\bibinfo  {journal} {Phys. Rev. A}\ }\textbf {\bibinfo {volume} {57}},\
  \bibinfo {pages} {4153--4175} (\bibinfo {year} {1998})}\BibitemShut {NoStop}%
\bibitem [{\citenamefont {{Barnum}}\ \emph {et~al.}(2000)\citenamefont
  {{Barnum}}, \citenamefont {{Knill}},\ and\ \citenamefont
  {{Nielsen}}}]{fidelities-capacities-2000}%
  \BibitemOpen
  \bibfield  {author} {\bibinfo {author} {\bibfnamefont {H.}~\bibnamefont
  {{Barnum}}}, \bibinfo {author} {\bibfnamefont {E.}~\bibnamefont {{Knill}}}, \
  and\ \bibinfo {author} {\bibfnamefont {M.~A.}\ \bibnamefont {{Nielsen}}},\
  }\bibfield  {title} {\enquote {\bibinfo {title} {On quantum fidelities and
  channel capacities},}\ }\href@noop {} {\bibfield  {journal} {\bibinfo
  {journal} {IEEE Trans. Inf. Theory}\ }\textbf {\bibinfo {volume} {46}},\
  \bibinfo {pages} {1317--1329} (\bibinfo {year} {2000})}\BibitemShut {NoStop}%
\bibitem [{\citenamefont {Lloyd}(1997)}]{Lloyd_capacity_97}%
  \BibitemOpen
  \bibfield  {author} {\bibinfo {author} {\bibfnamefont {Seth}\ \bibnamefont
  {Lloyd}},\ }\bibfield  {title} {\enquote {\bibinfo {title} {Capacity of the
  noisy quantum channel},}\ }\href@noop {} {\bibfield  {journal} {\bibinfo
  {journal} {Physical Review A}\ }\textbf {\bibinfo {volume} {55}},\ \bibinfo
  {pages} {1613--1622} (\bibinfo {year} {1997})}\BibitemShut {NoStop}%
\bibitem [{\citenamefont {Shor}(2002)}]{Shor_direct_capacity2002}%
  \BibitemOpen
  \bibfield  {author} {\bibinfo {author} {\bibfnamefont {P.~W.}\ \bibnamefont
  {Shor}},\ }\bibfield  {title} {\enquote {\bibinfo {title} {The quantum
  channel capacity and coherent information},}\ }in\ \href@noop {} {\emph
  {\bibinfo {booktitle} {Lecture Notes, MSRI Workshop on Quantum
  Computation}}}\ (\bibinfo {year} {2002})\BibitemShut {NoStop}%
\bibitem [{\citenamefont {{Devetak}}(2005)}]{Devetak-capacity-2005}%
  \BibitemOpen
  \bibfield  {author} {\bibinfo {author} {\bibfnamefont {I.}~\bibnamefont
  {{Devetak}}},\ }\bibfield  {title} {\enquote {\bibinfo {title} {The private
  classical capacity and quantum capacity of a quantum channel},}\ }\href@noop
  {} {\bibfield  {journal} {\bibinfo  {journal} {IEEE Trans. Inf. Theory}\
  }\textbf {\bibinfo {volume} {51}},\ \bibinfo {pages} {44--55} (\bibinfo
  {year} {2005})}\BibitemShut {NoStop}%
\bibitem [{\citenamefont {Devetak}\ and\ \citenamefont
  {Shor}(2005)}]{Devetak-Shor-2005}%
  \BibitemOpen
  \bibfield  {author} {\bibinfo {author} {\bibfnamefont {I.}~\bibnamefont
  {Devetak}}\ and\ \bibinfo {author} {\bibfnamefont {P.}~\bibnamefont {Shor}},\
  }\bibfield  {title} {\enquote {\bibinfo {title} {The capacity of a quantum
  channel for simultaneous transmission of classical and quantum
  information},}\ }\href@noop {} {\bibfield  {journal} {\bibinfo  {journal}
  {Commun. Math. Phys.}\ }\textbf {\bibinfo {volume} {256}},\ \bibinfo {pages}
  {287–303} (\bibinfo {year} {2005})}\BibitemShut {NoStop}%
\bibitem [{\citenamefont {Shor}()}]{c_capacity_Shor2004}%
  \BibitemOpen
  \bibfield  {author} {\bibinfo {author} {\bibfnamefont {P.W.}\ \bibnamefont
  {Shor}},\ }\bibfield  {title} {\enquote {\bibinfo {title} {The classical
  capacity achievable by a quantum channel assisted by limited entanglement},}\
  }\href@noop {} {\bibfield  {journal} {\bibinfo  {journal} {preprint (2004)}\
  }}\bibinfo {note} {ArXiv[quant-ph]:0402129}\BibitemShut {NoStop}%
\bibitem [{\citenamefont {Koashi}\ and\ \citenamefont {Imoto}(2001)}]{KI2001}%
  \BibitemOpen
  \bibfield  {author} {\bibinfo {author} {\bibfnamefont {M.}~\bibnamefont
  {Koashi}}\ and\ \bibinfo {author} {\bibfnamefont {N.}~\bibnamefont {Imoto}},\
  }\bibfield  {title} {\enquote {\bibinfo {title} {Compressibility of quantum
  mixed-state signals},}\ }\href@noop {} {\bibfield  {journal} {\bibinfo
  {journal} {Phys. Rev. Lett.}\ }\textbf {\bibinfo {volume} {87}},\ \bibinfo
  {pages} {017902} (\bibinfo {year} {2001})}\BibitemShut {NoStop}%
\bibitem [{\citenamefont {Koashi}\ and\ \citenamefont {Imoto}(2002)}]{KI2002}%
  \BibitemOpen
  \bibfield  {author} {\bibinfo {author} {\bibfnamefont {M.}~\bibnamefont
  {Koashi}}\ and\ \bibinfo {author} {\bibfnamefont {N.}~\bibnamefont {Imoto}},\
  }\bibfield  {title} {\enquote {\bibinfo {title} {Operations that do not
  disturb partially known quantum states},}\ }\href@noop {} {\bibfield
  {journal} {\bibinfo  {journal} {Phys. Rev. A}\ }\textbf {\bibinfo {volume}
  {66}},\ \bibinfo {pages} {022318} (\bibinfo {year} {2002})}\BibitemShut
  {NoStop}%
\bibitem [{\citenamefont {Hayden}\ \emph {et~al.}(2004)\citenamefont {Hayden},
  \citenamefont {Jozsa}, \citenamefont {Petz},\ and\ \citenamefont
  {Winter}}]{Hayden2004}%
  \BibitemOpen
  \bibfield  {author} {\bibinfo {author} {\bibfnamefont {P.}~\bibnamefont
  {Hayden}}, \bibinfo {author} {\bibfnamefont {R.}~\bibnamefont {Jozsa}},
  \bibinfo {author} {\bibfnamefont {D.}~\bibnamefont {Petz}}, \ and\ \bibinfo
  {author} {\bibfnamefont {A.}~\bibnamefont {Winter}},\ }\bibfield  {title}
  {\enquote {\bibinfo {title} {Structure of states which satisfy strong
  subadditivity of quantum entropy with equality},}\ }\href@noop {} {\bibfield
  {journal} {\bibinfo  {journal} {Commun. Math. Phys.}\ }\textbf {\bibinfo
  {volume} {246}},\ \bibinfo {pages} {359--374} (\bibinfo {year}
  {2004})}\BibitemShut {NoStop}%
\bibitem [{\citenamefont {Holevo}(1973)}]{Holevo1973}%
  \BibitemOpen
  \bibfield  {author} {\bibinfo {author} {\bibfnamefont {A.~S}\ \bibnamefont
  {Holevo}},\ }\bibfield  {title} {\enquote {\bibinfo {title} {Bounds for the
  quantity of information transmitted by a quantum communication channel},}\
  }\href@noop {} {\bibfield  {journal} {\bibinfo  {journal} {Probl. Inf.
  Transm.}\ }\textbf {\bibinfo {volume} {9}},\ \bibinfo {pages} {3--11}
  (\bibinfo {year} {1973})}\BibitemShut {NoStop}%
\bibitem [{\citenamefont {Fuchs}\ and\ \citenamefont
  {de~Graaf}(1999)}]{Fuchs1999}%
  \BibitemOpen
  \bibfield  {author} {\bibinfo {author} {\bibfnamefont {C.~A.}\ \bibnamefont
  {Fuchs}}\ and\ \bibinfo {author} {\bibfnamefont {J.~van}\ \bibnamefont
  {de~Graaf}},\ }\bibfield  {title} {\enquote {\bibinfo {title} {Cryptographic
  distinguishability measures for quantum-mechanical states},}\ }\href@noop {}
  {\bibfield  {journal} {\bibinfo  {journal} {IEEE Trans. Inf. Theory}\
  }\textbf {\bibinfo {volume} {45}},\ \bibinfo {pages} {1216--1227} (\bibinfo
  {year} {1999})}\BibitemShut {NoStop}%
\bibitem [{\citenamefont {Fannes}(1973)}]{Fannes1973}%
  \BibitemOpen
  \bibfield  {author} {\bibinfo {author} {\bibfnamefont {M.}~\bibnamefont
  {Fannes}},\ }\bibfield  {title} {\enquote {\bibinfo {title} {A continuity
  property of the entropy density for spin lattice systems},}\ }\href@noop {}
  {\bibfield  {journal} {\bibinfo  {journal} {Commun. Math. Phys.}\ }\textbf
  {\bibinfo {volume} {21}},\ \bibinfo {pages} {291--294} (\bibinfo {year}
  {1973})}\BibitemShut {NoStop}%
\bibitem [{\citenamefont {Audenaert}(2007)}]{Audenaert2007}%
  \BibitemOpen
  \bibfield  {author} {\bibinfo {author} {\bibfnamefont {K.~M.~R.}\
  \bibnamefont {Audenaert}},\ }\bibfield  {title} {\enquote {\bibinfo {title}
  {A sharp continuity estimate for the von neumann entropy},}\ }\href@noop {}
  {\bibfield  {journal} {\bibinfo  {journal} {J. Phys. A: Math. Theor.}\
  }\textbf {\bibinfo {volume} {40}},\ \bibinfo {pages} {8127--8136} (\bibinfo
  {year} {2007})}\BibitemShut {NoStop}%
\bibitem [{\citenamefont {Rockafeller}(1970)}]{Rockafeller}%
  \BibitemOpen
  \bibfield  {author} {\bibinfo {author} {\bibfnamefont {R.~T.}\ \bibnamefont
  {Rockafeller}},\ }\href@noop {} {\emph {\bibinfo {title} {Convex Analysis}}}\
  (\bibinfo  {publisher} {Princeton University Press},\ \bibinfo {year}
  {1970})\BibitemShut {NoStop}%
\bibitem [{\citenamefont {Csisz\'ar}\ and\ \citenamefont {K\"orner}(2nd~ed.
  2011)}]{csiszar_korner_2011}%
  \BibitemOpen
  \bibfield  {author} {\bibinfo {author} {\bibfnamefont {I.}~\bibnamefont
  {Csisz\'ar}}\ and\ \bibinfo {author} {\bibfnamefont {J.}~\bibnamefont
  {K\"orner}},\ }\href@noop {} {\emph {\bibinfo {title} {Information Theory:
  Coding Theorems for Discrete Memoryless Systems}}}\ (\bibinfo  {publisher}
  {Cambridge Univ. Press},\ \bibinfo {year} {2nd~ed. 2011})\BibitemShut
  {NoStop}%
\end{thebibliography}%

\end{document}